%% file: main.tex
\documentclass[10pt,twocolumn,twoside]{IEEEtran}
\usepackage[utf8]{inputenc} 
\usepackage[T1]{fontenc}
\usepackage{url}
\usepackage{ifthen}
\usepackage{cite}
\usepackage{color}
\usepackage{xcolor}
\usepackage{graphicx}
\usepackage{amssymb}
\usepackage{subcaption}
\usepackage{bbm}
\usepackage{enumitem}
\usepackage{xfrac}
\usepackage[cmex10]{amsmath} 

\usepackage{thmtools}
\usepackage{thm-restate}

\usepackage{pgfplots}
\usepgfplotslibrary{groupplots,dateplot}
\usetikzlibrary{patterns,shapes.arrows}
\pgfplotsset{compat=newest}
\usepackage{tikzscale}
\usepackage[normalem]{ulem}
\usepackage{amsthm}
\usepackage[final]{hyperref} 
\usepackage{cleveref}
\usepackage{bbm}
\hypersetup{
	colorlinks=true,       
	linkcolor=blue,        
	citecolor=blue,        
	filecolor=magenta,     
	urlcolor=blue         
}

\usepackage{balance}
\newtheorem{theorem}{Theorem}
\newtheorem{lemma}{Lemma}
\newtheorem{corollary}{Corollary}
\newtheorem{definition}{Definition}

 \usepackage{setspace}
\allowdisplaybreaks

\newif\ifnotes
\notestrue
\newcommand{\snote}[1]{\ifnotes{{\sf\color{blue} [Sundar: #1]}}\fi}

\newcommand\numberthis{\addtocounter{equation}{1}\tag{\theequation}}

\newcommand{\priors}{\mathbf p}
\newcommand{\prior}{p}
\newcommand{\testmatrix}{G}
\newcommand{\defectives}{\mathcal D}
\newcommand{\infectionstatuses}{\mathbf U}
\newcommand{\infectionstatus}{U}
\newcommand{\numtests}{T}
\newcommand{\numitems}{N}
\newcommand{\testresults}{\mathbf y}
\newcommand{\testResult}{\mathit{y}}
\newcommand{\perror}{\mathbb P_{err}}
\newcommand{\mapdecoder}{R_{map}}
\newcommand{\init}{\mathrm{init}}
\newcommand{\errorMapEvent}{\mathcal{E}}

\usepackage{mathtools}

\DeclareMathOperator*{\expect}{\mathbb E}
\DeclareMathOperator*{\argmax}{arg\,max}

\DeclareMathOperator*{\mean}{mean}

\newcommand{\binEntropy}{\mathit{h_2}}

\newcommand{\susceptState}{\mathcal{S}}
\newcommand{\infectState}{\mathcal{I}}
\newcommand{\recState}{\mathcal{R}}
\newcommand{\SIRmemberState}{\mathit{X}}
\newcommand{\SIRmemberStates}{\mathbf{X}}
\newcommand{\communitySize}{\mathit{C}}

\newcommand{\infectProb}{\mathit{q}}
\newcommand{\recoverProb}{\mathit{r}}

\renewcommand{\Pr}{\mbox{${\mathbb P}$} }

\newcommand{\graph}{\mathcal{G}}

\newcommand{\nDef}{\mathit{k}}

\newcommand{\memberState}{\mathit{U}}

\newcommand{\nTests}{\mathit{T}}

\newcommand{\transmissibility}{\mathit{\beta}}
\newcommand{\recovery}{\mathit{\gamma}}

\newcommand{\testIndex}{\mathit{\tau}}

\newcommand{\defVariable}{\mathit{U}}

\begin{document}

\title{Dynamic group testing to control and monitor disease progression in a population}

\author{%
\vspace{0.5cm}
  Sundara Rajan Srinivasavaradhan\IEEEauthorrefmark{1},
                    Pavlos Nikolopoulos\IEEEauthorrefmark{2},
                    Christina Fragouli\IEEEauthorrefmark{1},
                    Suhas Diggavi\IEEEauthorrefmark{1}\\
  \IEEEauthorrefmark{1}%
                    University of California, Los Angeles, Electrical and Computer Engineering,\\ email: \{sundar, christina.fragouli, suhasdiggavi\}@ucla.edu \\
  \IEEEauthorrefmark{2}%
                    EPFL, Switzerland, email: pavlos.nikolopoulos@epfl.ch
}

\maketitle

\begin{abstract}
In the context of a pandemic like COVID-19, and until most people are vaccinated,
proactive testing and interventions have been proved to be the only means to contain the disease spread.
Recent academic work has offered significant evidence in this regard, but a critical question is still open: Can we accurately identify all new infections that happen every day, without this being forbiddingly expensive, i.e., using only a fraction of the tests needed to test everyone everyday (complete testing)?

Group testing offers a powerful toolset for minimizing the number of tests, but it does not account for the time dynamics behind the infections. 
Moreover, it typically assumes that people are infected independently, while infections are governed by community spread.
Epidemiology, on the other hand, does explore time dynamics and community correlations through the well-established continuous-time SIR stochastic network model, but the standard model does not incorporate discrete-time testing and interventions. 

In this paper, we introduce a ``discrete-time SIR stochastic block model'' that also allows for group testing and interventions on a daily basis. 
Our model can be regarded as a discrete version of the continuous-time SIR stochastic network model over a specific type of weighted graph that captures the underlying community structure.
We analyze that model w.r.t.\ the minimum number of group tests needed everyday to identify all infections with vanishing error probability. 
We find that one can leverage the knowledge of the community and the model to inform nonadaptive group testing algorithms that are order-optimal,
and therefore achieve the same performance as complete testing using a much smaller number of tests.
\end{abstract}


\IEEEpeerreviewmaketitle

\input{introduction}
\input{related}
\input{model}
\input{results}

\input{evaluation}
\section{Acknowledgements}
This work was supported in part by NSF grants \#2007714, \#1705077 and UC-NL grant LFR-18-548554. We also thank Katerina Argyraki for her ongoing support and the valuable discussions we have had about this project.

\bibliographystyle{IEEEtran}
\bibliography{bibliography}

\appendices
\input{appendix}



\end{document}

%% file: introduction.tex
\section{Introduction}

	COVID-19 has revealed the key role of
	accurate epidemiological models and testing in the fight against pandemics\cite{art1,art2,art4,Cov-GpTest-1,Cov-GpTest-2,kucirka2020-PCR}. 
	For any new disease or variant of the existing ones, we will always need to fast develop strategies that allow efficient testing of populations and empower targeted interventions.
	This poses several daunting challenges: (i) we need to test populations not just once but in a continual manner (on a daily basis), 
	and (ii) we need to estimate the epidemic state of each individual and isolate \textit{only} the infected ones.
	And all these, under the accuracy and cost limitations imposed by the various types of tests.
	
	Recent works have identified the significance of proactive testing and individual-level intervention for the control of the disease spread 
	(e.g. \cite{taipale2021populationscale,taipale2020population,bergstrom2020frequency}), 
	but 
	to the best of our knowledge 
	none of them addresses the challenges above efficiently. 
	Most solutions rely on the idea of "testing everyone individually", which is inefficient for two reasons:
	on one hand, using cheap rapid testing usually results in many people (false positives) ending up in isolation without reason and at non-negligible societal cost;
	on the other hand, using accurate tests like PCR can be forbiddingly expensive. 
	As a result, these works need to either neglect the cost of the former or alleviate the cost of the latter by scheduling tests on a (bi)weekly or monthly basis.
	
	Therefore, a critical question is still open: can we use accurate/expensive tests more efficiently? 
	In other words, can we identify all new infections that happen each day (complete testing performance), using significantly fewer accurate/expensive tests than complete testing? Complete accurate testing  (e.g. PCR) on a daily basis and isolation of infected individuals can significantly reduce the number of infected people, as the example in Fig.~\ref{fig:testing_vs_notesting} illustrates. Note  that even with complete testing  new infections still occur  due to the delay between testing and receiving the test results (Fig.~\ref{fig:testing_vs_notesting} assumes the usual delay of one day).
	Still, this is the  best performance we can hope for, both in terms of containing infection and alleviating the societal impact of "false" quarantines; we thus ask how many tests do we really need to replicate it.

	Traditional group testing strategies offer a powerful toolset for minimizing the number of tests, but they do not account for the time dynamics of a disease spread and do not take into account community structure.  
	When the number of available tests is limited, two strategies are usually applied: 
	sample testing, which tests only a sample of selected individuals, 
	and/or group testing, which pools together diagnostic samples to reduce the number of tests needed to identify infected individuals in a population (e.g., see~\cite{GroupTestingMonograph} and references therein).	Both examine a static scenario: 
	the state of individuals is fixed (infected or not),
	and the goal is to identify all infected ones. 
	
	To the best of our knowledge, our work in \cite{isit-paper} was the first paper that targeted community-aware, group-test design for the dynamic case. 
	In that work we used the well-established continuous-time SIR stochastic network model in \cite{mathOfEpidemicsOnNetworks}, where individuals are regarded as the vertices of a graph $\graph$ and an edge denotes a contact between neighboring vertices, 
	 and explored group testing strategies that were able to track the epidemic state evolution at an individual level, using a small number of tests.
	However, due to the complexity of the continuous time model, we were not able to provide theoretical guarantees for the minimum number of tests, and although we did consider testing delays, we left interventions for future work.  
	
	In this paper, we allow interventions, 
	we use discrete-time SIR models for disease spread and we derive theoretical guarantees.  Discrete time models fit more naturally with testing and intervention (which happen at discrete time-intervals), and are more amenable to analysis
	enabling methods to derive guarantees on the number of tests needed to achieve close-to-complete-testing accuracy. In this paper we use a model called the ``discrete-time SIR stochastic block model,''
		which can be considered as a discrete version of the continuous-time SIR stochastic network model over a specific type of weighted graph.
		The graph used captures knowledge of an underlying community structure, as discussed in Section~\ref{sec:setup}. {In Appendix~\ref{app:discrete_vs_cont}, we compare the continuous-time model from \cite{mathOfEpidemicsOnNetworks} with the discrete-time model introduced in our work and justify the use of our discrete-time model.}
We also note that our results   are applicable to a  larger set of SIR models, as  discussed in Section~\ref{subsec:dynamic_testing_results}.

	Our main conclusion is that we can leverage the knowledge of the community and the dynamic model to inform  group testing algorithms that are order-optimal and use a much smaller number of tests than complete testing to achieve the same performance. We arrive at this conclusion building on the following contributions.
	
	
 We first argue that for  discrete-time SIR models, given test results that identify the infection state the  previous day, the problem of identifying the new infections each day reduces to static-case group testing  with independent (but not identical) priors. 
		So, existing nonadaptive algorithms such as CCA and/or random testing~\cite{prior} can be reused. \Cref{fig:TI_chain} illustrates the sequence of events taking place on each day $t$.

		 We then derive a new lower bound (\Cref{thm:static_lower_bound}) for the number of tests needed in the case of independent (but not identical) priors. The main benefit of the new bound is not on "improving" upon the well-known entropy lower bound (stated as \Cref{lemma:entropy_LB}), but having a form that allows to prove order-optimality of group testing algorithms.
		 In particular, we can prove that under mild assumptions existing nonadaptive algorithms are order-optimal in the static case (\Cref{cor:order_bounded_priors}).
		This, in our opinion, is an interesting result on its own, since non-identical priors static group testing remains a relatively unexplored field compared to i.i.d. probabilistic group testing. 

Finally, we derive conditions on the 
discrete-time SIR stochastic block model parameters under which order-optimal group test designs for the static case are also optimal for the dynamic case (\Cref{thm:bounded_proof}).  Simulation results show that indeed under these conditions we can achieve the performance of complete individual testing using a much smaller (close to the entropy lower bound) number of tests; for example, over a period of $50$ days, group testing needs an average of around $100$ tests per day for a population of $1000$ individuals.
Our simulations use existing non-adaptive test designs - we do not derive new code designs as the existing ones are sufficient. However, we do use marginal probabilities derived daily from the SIR model to inform the group design: that is, the group tests we use vary from day to day, and their design leverages the knowledge of the underlying system dynamics that depend on the community structure,  as well as the  previous day test results. 

The rest of the paper is organized as follows: Section~\ref{sec:related} discusses related work; Section~\ref{sec:setup} provides our setup and background; Sec~\ref{sec:results} contains our main results; Section~\ref{sec:numerical} provides numerical evaluation and concludes the paper.

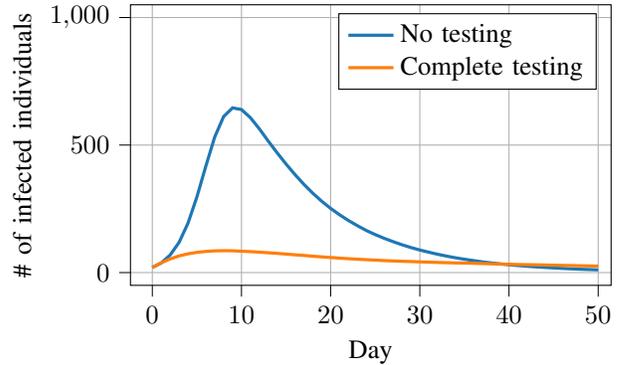
\begin{figure}
    \centering
    \input{Figures/Disease_progression}
    \caption{\small Discrete-time SIR stochastic block model simulated on a population of 1000 individuals. Notice that without any testing or intervention a large fraction of the population gets infected. With \textit{complete testing} (individually testing everyone everyday) and intervention (isolating individuals who are identified as infected) we can flatten the curve to a large extent. We assume that test results are only available the next day; if the test results were instantaneous we can identify all infections on the first day and isolate them and there would be no subsequent new infections.}
    \label{fig:testing_vs_notesting}
\end{figure}

\begin{figure*}[ht]
	\centering
	\includegraphics[width = 0.8\textwidth]{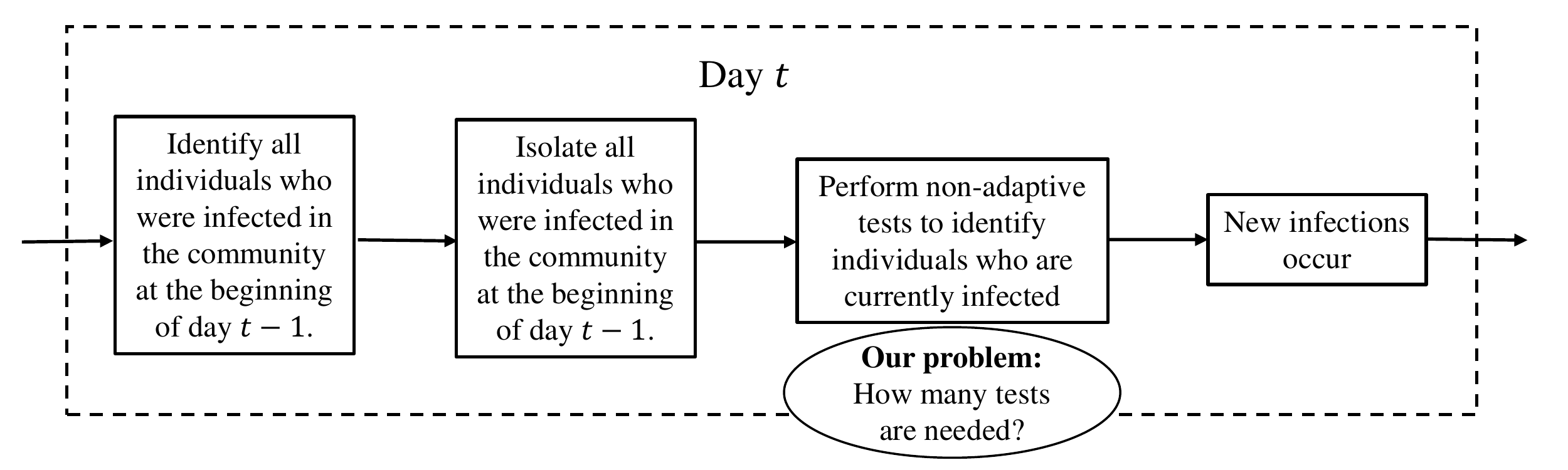}
	\caption{\small The dynamic testing problem with daily interventions. How many tests are needed to achieve complete testing performance everyday, given that test results become available after a day's delay.} 
	\label{fig:TI_chain}
\end{figure*}

%% file: Figures/Disease_progression.tex
\begin{tikzpicture}

\definecolor{color0}{rgb}{0.12156862745098,0.466666666666667,0.705882352941177}
\definecolor{color1}{rgb}{1,0.498039215686275,0.0549019607843137}
\definecolor{color2}{rgb}{0.172549019607843,0.627450980392157,0.172549019607843}
\definecolor{color3}{rgb}{0.83921568627451,0.152941176470588,0.156862745098039}
\definecolor{color4}{rgb}{0.580392156862745,0.403921568627451,0.741176470588235}
\definecolor{color5}{rgb}{0.549019607843137,0.337254901960784,0.294117647058824}

\begin{axis}[
width = 0.9\columnwidth,
height = 0.6\columnwidth,
legend cell align={left},
legend pos=north east,
tick align=outside,
tick pos=left,
x grid style={white!69.0196078431373!black},
xmajorgrids,
xmin=-2.45, xmax=51.45,
xminorgrids,
xlabel = Day,
ylabel = \# of infected individuals,
xtick style={color=black},
y grid style={white!69.0196078431373!black},
ymajorgrids,
ymin=-50, ymax=1050,
yminorgrids,
ytick style={color=black}
]
\addplot [very thick, color0]
table {%
0 19.805
1 37.77
2 68.92
3 118.425
4 193.635
5 296.72
6 416.98
7 531.775
8 611.995
9 645.905
10 639.21
11 607.295
12 563.775
13 515.625
14 469.57
15 425.42
16 384.31
17 345.925
18 310.98
19 279.67
20 251.585
21 226.67
22 204.18
23 183.875
24 165.025
25 148.895
26 134.365
27 121.295
28 109.11
29 98.17
30 88.48
31 79.915
32 71.84
33 64.545
34 58.015
35 52.32
36 47.08
37 42.1
38 37.895
39 34.08
40 30.745
41 27.605
42 24.87
43 22.48
44 20.3
45 18.255
46 16.325
47 14.61
48 13.19
49 11.925
50 10.685
};
\addlegendentry{No testing}
\addplot [very thick, color1]
table {%
0 20.08
1 38.435
2 53.38
3 65.31
4 73.525
5 79.275
6 82.955
7 85.19
8 85.84
9 85.47
10 83.825
11 82.21
12 80.025
13 77.725
14 75.375
15 72.555
16 69.77
17 67.16
18 64.19
19 61.59
20 59.13
21 56.655
22 54.44
23 52.395
24 50.595
25 48.86
26 46.995
27 45.645
28 44.295
29 43.305
30 42.08
31 40.86
32 39.74
33 38.775
34 37.825
35 37.34
36 36.545
37 35.405
38 34.36
39 33.555
40 32.505
41 31.705
42 31.14
43 30.615
44 29.88
45 29.035
46 28.17
47 27.53
48 26.845
49 25.995
50 25.495
};
\addlegendentry{Complete testing}
\end{axis}

\end{tikzpicture}

%% file: related.tex
\section{Related Work}
\label{sec:related}



As stated in our introduction, our work shares similar goals with our prior work in \cite{isit-paper}, where we considered the well established continuous-time SIR stochastic network model (see \cite{mathOfEpidemicsOnNetworks}) and focused on how many tests to use and whom to test in order to track the infected individuals in the population. 
That work also explored how well one can learn the infected individuals given delayed test results, but gave no theoretical guarantees on the methods and did not consider intervention.
Our discrete-time model for disease spread in this work, however, is more amenable to analysis and illustrates better the usefulness of group testing, being at the same time useful for practical reasons (more about this in \Cref{sec:setup}). 
We further note that our results are applicable to a more general set of SIR models as discussed in \Cref{subsec:dynamic_testing_results}, remark 2.

Our model is closely related to the independent cascade model (see for example \cite{kempe2003maximizing} and references therein), studied in the context of influence maximization in social networks, where we can interpret influence/rumor propagation as infections in our context. A crucial difference of our model from this is that our model allows multiple opportunities of infections over time whereas the independent cascade model only allows one opportunity to "infect". Therefore, as is noted, in our model the infectious individuals remain infectious until recovered or isolated. 



The work in \cite{goenka2020contact} considers a discrete stochastic model for the progression of COVID-19 based on contact networks 
and leverages  the model dynamics to inform a group test decoder; 
however their scope is different, as  they test infrequently and thus infections are highly correlated,  do not consider interventions,
do not look for optimal group test designs, and do not provide theoretical guarantees on the number of tests needed.

Since we use the main principles of the SIR model  our work is closely related to epidemic modeling. 
Works in epidemiology discuss the implications of testing and intervention for COVID-19 employing stochastic network models (see \cite{bergstrom2020frequency,taipale2020population} and references therein) but do not consider test designs that exploit the knowledge of the underlying dynamical system. 
Works in control theory (see \cite{molnar2020safety} and references therein)  consider deterministic SIR compartment models (at the population level) and focus on intervention schemes. Here we are interested in both testing and intervention and use an individual-level SIR model.

Our work can be positioned in the general context of community-aware group testing where infected are not independent, and correlations follow from the community structure. 
Our work in \cite{GroupTesting-community-nonOverlap,GroupTesting-community-overlap} demonstrated that using a known community structure to design group testing strategies and decoding, can significantly extend the advantages of group testing by utilizing these structural dependencies.  Concurrently, the  works in \cite{zhu2020noisy, goenka2020contact} proposed decoding algorithms that take the community structure into account.
Following up on these works in the static case (without temporal dynamics), there have been other recent works with similar goals \cite{ayfer2021adaptive, sennur2021group, bertolotti2020network}. Our work also leverages a known community structure that informs the system dynamics as well as the group test designs. 

Further related to static group testing is the work on graph-constrained group testing (see for example \cite{cheraghchi2012graph}, \cite{karbasi2012sequential}), which solves the problem of how to design group tests when there are constraints on which samples can be pooled together, provided in the form of a graph. In our case, no such constraints exist and individuals can be pooled together into tests freely.

%% file: model.tex
\section{Preliminaries and problem formulation}
\label{sec:setup}
In this section we formalize our setup. 
Since our work for the dynamic case builds upon existing ones from static group testing, we first review some major results in that area that we also reuse in our paper (\Cref{sec:tradGP}).
We then provide our model (\Cref{subsec:dynamic_testing_problem}) and problem formulation (\Cref{subsec:dynamic_testing_problem}).

\subsection{Preliminary: review of results from static group testing}
\label{sec:tradGP}
Traditional group testing typically assumes a population of $\numitems$ individuals out of which some are infected. 
Three infection models are typically considered: (i) in the {\em combinatorial priors model},  a fixed number of infected individuals $\nDef$, are selected uniformly at random among all sets of size $\nDef$; (ii) in \textit{i.i.d probabilistic priors model}, each individual is i.i.d infected with probability $p$;
(iii) in the {\em non-identical probabilistic priors model}, each item $i$ is infected independently of all others with prior probability $\prior_i$, so that the expected number of infected members is $\bar{\nDef} = \sum_{i=1}^{\numitems}\prior_i$~\cite{prior}. In this paper we mostly use results that apply to the last case.

A group test $\testIndex$ takes as input samples from $n_\testIndex$ individuals, pools them together and  outputs a single value:  positive if any one of the samples is infected, and negative if none is infected.  
More precisely, let ${\infectionstatus_i=1}$  when individual $i$ is infected and $0$ otherwise. 
Then the group testing output $\testResult_\testIndex$ takes a binary value calculated as $\testResult_\testIndex= \bigvee_{i\in \defectives_{\testIndex}} \defVariable_i$\footnote{We assume that the tests are noiseless here, for simplicity. The group testing literature also extensively studies the case when the testing output is noisy.}, 
where $\bigvee$ stands for the \texttt{OR} operator (disjunction) and $\defectives_{\testIndex}$ is the group of people participating in the test. 

The usual goal in static group testing is to design a testing algorithm that is able to identify all infection statuses $\infectionstatuses = \left( \infectionstatus_1,\ldots,\infectionstatus_\numitems\right)$. 
These algorithms can be adaptive or non-adaptive. 
Adaptive testing uses the outcome of previous tests to decide what tests to perform next. An example of adaptive testing is {\em binary splitting}, which implements a form of binary search.
Non-adaptive testing constructs, in advance, a \emph{test matrix} $\testmatrix\in \{0,1\}^{\nTests\times \numitems}$  where each row corresponds to one test, each column to one member, and the non-zero elements determine the set $\defectives_{\testIndex}$. 
Although adaptive testing uses less tests than non-adaptive, non-adaptive testing is often more practical as all tests can be executed in parallel.

The main challenge in static group testing is the number of group tests $\nTests=\nTests(\numitems)$ needed to identify the infected members without error or with high probability.  
In the following, we present some well established results that we reuse in our work 
(hereafter we use typical asymptotic notations; i.e., $f(n) = O(g(n)$, $f(n) =\Omega(g(n))$, or $f(n) =\Theta(g(n))$ 
respectively means that a $f(n) \leq c g(n)$, $f(n) \geq c g(n)$, $c_1 g(n) \leq f(n) \leq c_2 g(n)$ asymptotically, where $c$, $c_1$, $c_2$ are universal constants):



$\bullet$
For the probabilistic model (ii),
any non-adaptive algorithm with a success probability bounded away from zero as $\numitems \rightarrow \infty$ must have $\nTests = \Omega \left(\min\{ \bar{\nDef} \log{\numitems},\numitems\}\right)$~\cite[Theorem 1]{bay2020optimal},\cite{coja-oghlan20a}.
This means that either any non-adaptive group testing with a number of tests $O(\bar{\nDef} \log{N})$ is order optimal, or individual testing is order optimal\footnote{The achievability and converse results provided here are usually proved for combinatorial model (i) (a summary can be found in~\cite{price2020fast}), but they are directly applicable to model (ii) by considering $p=\nDef/\numitems$ (see Theorem 1.7 and Theorem 1.8 in \cite{GroupTestingMonograph} or~\cite{bay2020optimal}).}.
In particular, random test designs, such as i.i.d. Bernoulli~\cite{bernoulli_testing1,bernoulli_testing2,PhaseTrans-SODA16} and near-constant tests-per-item~\cite{coja-oghlan19,ncc-Johnson} have been proved to be order-optimal in a sparse regime where $\bar{\nDef} = \Theta(\numitems^\alpha)$ and $\alpha \in (0,1)$. 
In fact, in the same regime,~\cite{coja-oghlan20a} has provided the precise constants for optimal non-adaptive group testing.
Conversely, classic individual testing has been proved to be optimal in the linear  ($\bar{\nDef} = \Theta(\numitems)$)~\cite{individual-optimal} and the mildy sublinear regime ($\bar{\nDef} = \omega(\frac{\numitems}{\log\numitems})$)~\cite{bay2020optimal}.

$\bullet$ 
For the probabilistic model (iii),  
a lower bound for the number of tests needed is given by the entropy, stated below: 
\begin{lemma}[Entropy lower bound]
\label{lemma:entropy_LB}
Consider the non-identical probabilistic priors model of static group testing, where each individual $i\in [\numitems]$ is infected independently with probability $p_i$. The number of tests $T$  needed by a non-adaptive algorithm to identify the infection status of all individuals with a vanishing probability of error satisfies
$$T\geq \sum_{i=1}^{\numitems} \binEntropy\left( \prior_i \right),$$
where $\binEntropy(\cdot)$ is the binary entropy function.
\end{lemma}
See Appendix A in \cite{prior} for a proof.
On the algorithmic side, two known algorithms are: 
the adaptive laminar algorithms that need at most $2 \sum_{i=1}^{\numitems} \binEntropy\left( \prior_i \right) + 2 \bar{\nDef}$ tests on average,
and the ``Coupon collector'' nonadaptive algorithm (CCA) that needs at most $\nTests \leq 4 e (1+\delta) \bar{\nDef} \ln{\numitems}$ test to achieve an error probability no larger than $2 N^{-\delta}$ whenever $\prior_i \leq \sfrac{1}{2}$ ~\cite{prior,Nonadaptive-1}.\\

\noindent \textbf{Distinction from traditional group testing.} In this paper, we focus on probabilistic infections and non-adaptive test designs, but we differ from traditional group testing in two ways:\\
(a) Traditional  group testing examines a static scenario, where the state of individuals is fixed (infected or not); we are instead interested in a dynamic scenario, where the state of an individual may change, even during the test period. This is particularly true since test results may not be available instantaneously but instead with a delay (e.g. after one day).\\
(b) The infection probabilities $\prior_i$ are not independent; instead, they are correlated where the correlation is induced by the underlying community structure and dynamic infection spread model we consider (for instance, two individuals who live in the same household are more likely to be both infected or not). This implies that $U_i$ and $U_j$ are not independent, as is the assumption in traditional group testing. 

\subsection{Discrete-time SIR stochastic block model}
\label{subsec:model}
We now describe our infection model via the discrete-time SIR stochastic block model with parameters $(\numitems,C,q_1,q_2,p_{\init})$. Consider a population of size $\numitems$ that is partitioned into multiple communities of size $\communitySize$. For simplicity we assume that $\sfrac{\numitems}{\communitySize}$ is an integer.
On any day $t\in \mathbb N$, each individual can be in one of three states: Susceptible ($\susceptState$), Infected ($\infectState$) or Recovered ($\recState$). 
Let $\SIRmemberState_i^{(t)}\in \{\susceptState,\infectState,\recState\}$ denote the state of individual $i$ on day $t$, and define the state of the system as ${\SIRmemberStates}^{(t)} \triangleq (\SIRmemberState_1^{(t)},\SIRmemberState_2^{(t)},...,\SIRmemberState_N^{(t)}).$ A small number of individuals are initially infected, and all new infections occur  during ``transmissible contacts'' between infected and susceptible individuals. Recoveries occur independent of infections.


 More precisely, on day $t=0$, every individual is i.i.d infected with probability $p_{\init}$. The following steps repeat everyday starting at $t=1$:
\begin{itemize}[leftmargin = 3mm]
    \item  An infected individual in some community infects a susceptible one from the same community w.p.~$\infectProb_1$, independently of the other infected individuals of the community.
    \item An infected individual in some community infects a susceptible one from another community w.p.~$\infectProb_2$, independently from all other infected individuals.
    \item An infected individual recovers independently from all other individuals w.p.~$\recoverProb$.
\end{itemize}

The discrete-time SIR stochastic block model can be envisioned as a discrete version of the well-established continuous-time SIR  stochastic network model~\cite{mathOfEpidemicsOnNetworks} on the corresponding weighted graph.
It inherits the main properties from the latter; for example, infections are transmitted only from an infected to a susceptible individual and both infections and recoveries are stochastic. 
The main difference is that in the continuous-time one, the infections and the recoveries happen according to continuous-time Markovian process with transmissibility rate $\transmissibility$ and recovery rate $\recovery$, which means that the time until a new state transition ($\susceptState \rightarrow \infectState$ or $\infectState \rightarrow \recState$) is exponentially distributed (with mean $\transmissibility$ or $\recovery$ respectively). Indeed this makes the event that an individual got infected and subsequently recovered within a single day possible in the continuous-time model, whereas this is impossible in our discrete-time model.

Learning the intra-community and inter-community structure  to model infection transmissions is, we believe,  also practically feasible. Close contact ``community'' data is often readily available; for example students in each classroom in a  school could form a community, and so could workers in the same office space.  We also note that community-level network models  alleviate some of the privacy concerns associated with using  contact tracing data which tracks the exact pairs of individuals who come in contact with each other on a daily basis.

A useful remark about our model is that the state of an individual $\SIRmemberState_i^{(t)} \in \{\susceptState,\infectState,\recState\}$ is different from the infection state $\memberState_i^{(t)} \in \{0,1\}$, where $1$ (resp. $0$) corresponds to the ``infected'' (resp. ``not infected''). Indeed $\infectionstatus_i^{(t)}=1$ iff $\SIRmemberState_i^{(t)}=\mathcal \infectState$. 
This difference is important in our context, because our tests do not distinguish between susceptible and recovered individuals.
In the remainder of the paper, $\SIRmemberState_i^{(t)}$ will be called the ``SIR state'' of individual $i$, 
while $\infectionstatus_i^{(t)}$ will be $i$'s ``infection status.''
As a result, 
whether a individual is infected or not changes with the day, 
and thus we now consider a random variable $\infectionstatus_i^{(t)}$ associated with each individual that describes whether it is infected on day $t$ ($t\in \mathbb \mathbb{N}$). 

\subsection{The dynamic testing problem formulation}
\label{subsec:dynamic_testing_problem}

As can be seen from Fig.~\ref{fig:testing_vs_notesting}, testing everyone, everyday, and isolating infected individuals helps drastically reduce the number of infections that happen on a given day. We assume that the results of a test administered on a particular day are available only the next day (as usually is the case with classic PCR testing for SARS-COV-2). We also isolate only the individuals who test positive, and we do so as soon as the test results are available. Moreover, we bring back the isolated individual into the population only when they are completely recovered. Note that in the SIR model, recovered individuals cannot get infected and play no role in transmitting the infection. Therefore, without loss of generality, we could assume that isolated individuals remain isolated for the rest of the testing period.

Given these assumptions, the question we ask is if complete testing is necessary to identify all new infections everyday, or if we can achieve the same performance as complete testing with significantly fewer number of tests. In particular, how many non-adaptive group tests are necessary and sufficient to identify all new infections (with a vanishing error probability) on each day? Our problem formulation is depicted in Fig.~\ref{fig:TI_chain}.



\noindent To aid a precise mathematical formulation for the problem, we first introduce some notation.

\begin{itemize}[leftmargin=3mm]
	\item $I^{(t)}_j$: number of new infections in community $j$ that occurred on day $t$. Note that in the set-up of Fig.~\ref{fig:TI_chain}, this number is also equal to the number of infected individuals remaining in community $j$ after intervention has been decided for day $t+1$. The new infections which happened on day $t$ will only be identified by the tests administered on day $t+1$, whose results are available only on day $t+2$.
	\item $p_j^{(t)}$: the probability of an individual in community $j$ who was susceptible at the end of day $t-1$ getting infected on day $t$. Note that this is same for every such individual in community $j$, by symmetry of the model. Moreover, we can calculate this probability as
	\begin{align*}
	    p_j^{(t)}&=1-\Pr(\text{individual is not infected on day } t)\\
	    &=1-(1-q_1)^{I_j^{(t-1)}}(1-q_2)^{\sum_{j'\neq j}I^{(t-1)}_{j'}}.
	\end{align*}
\end{itemize}

\begin{figure*}
    \centering
    \includegraphics[width=0.8\textwidth]{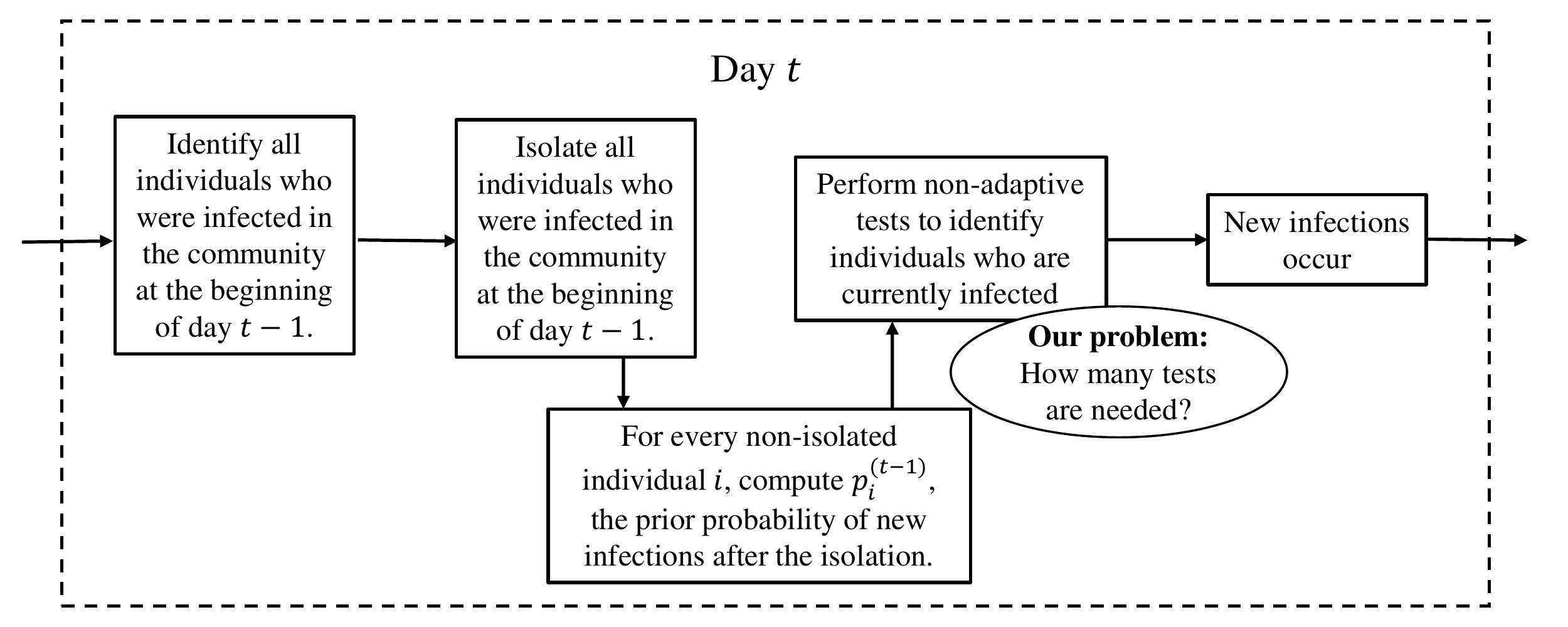}
    \caption{\small From dynamic to static testing: on day $t$ we perfectly learn the states of all non-isolated individuals at the time of testing on the previous day $t-1$. Given this information, we know that each susceptible individual in community $j$ is later infected with probability $p_j^{(t-1)}$ independent of every other individual. How many tests are needed to attain a vanishing probability of error on this non-identical static group testing problem?}
    \label{fig:problem_form}
\end{figure*}

\noindent{\bf Reduction to static group testing with non-identical probabilistic priors.}
Note that given $I_j^{(t-1)}\ \forall j$, an individual belonging to community $j$ is infected \textit{independently} of every other individual with probability $p_j^{(t)}$ on day $t$. Thus, conditioned on the infection status of all individuals on day $t-1$, the infections which happen on day $t$ are independent (but not identically distributed). Now in our dynamic testing problem set-up, on day $t$ we perfectly learn the infection statuses of all non-isolated individuals  at the time of testing on the previous day $t-1$. Given this information, we can exactly calculate the $p_j^{(t-1)}\ \forall j$ (see Fig.~\ref{fig:problem_form}), i.e. the probability that each susceptible individual in community $j$ was later infected because of the non-isolated infected individuals.
 
So, given accurate test results, the dynamic testing problem is transformed daily to the problem of static group testing with non identical probabilistic priors (model (iii) in \Cref{sec:tradGP}).
Therefore, the precise question we are after is the following: 
given that each individual in community $j$ is infected with probability $p_j^{(t)}$ independently of every other individual, how many tests are necessary and sufficient to learn the infection status with a vanishing probability of error? We answer this question in the next section. 


%% file: results.tex
\section{Main results}
\label{sec:results}

In this section, we prove our main theoretical results.
For brevity, we will use the terms ``i.i.d. priors'' and  ``non-identical priors'' to refer to i.i.d probabilistic priors model (ii)
and non identical probabilistic priors model (iii) from \Cref{sec:tradGP}, respectively.
The contents of this section are ordered as follows:
\begin{itemize}[leftmargin = 3mm]
    \item First, we provide a new lower bound on the number of tests required for the problem of static group testing with non i.i.d probabilistic priors (\Cref{thm:static_lower_bound}). 
    To prove \Cref{thm:static_lower_bound}, we use two intermediate results:
    (a) we show that any test design that ``works'' for a given prior probabilities of infection $(p_1,p_2,...,p_\numitems)$ also works for the \textit{reduced} prior probabilities $(p_1',p_2',...,p_N')$ where $p_i'\leq p_i \leq 0.5\ \forall i$. In words, we essentially prove that group testing is easier when the infections are sparser (\Cref{thm:static_reduced});
    and (b) we show the following interesting property of the optimal decoder (\Cref{lem:map_decoder_increasing}) -- if the optimal decoder correctly infers all the infection statuses when a set $\defectives$ is the set of infected individuals, then it will also correctly infer all the infection statuses when $\defectives' \subset \defectives$ is the set of infected individuals.
    
    \item Second, we use simple asymptotic arguments to show that some existing group testing strategies  (such as CCA~\cite{prior} for non-identical priors and random testing for i.i.d priors) are  order-optimal  for non-identical priors (Corollary~\ref{cor:order_bounded_priors}), when $p_{\max}=O(p_{\min})$, where $p_{\max}$ is the maximum entry in $(p_1,p_2,...,p_\numitems)$ and $p_{\min}$ is the minimum entry. The order $O(\cdot)$ is order with respect to the size of the population $\numitems$.
    
    \item Finally, in Theorem~\ref{thm:bounded_proof}, we bridge the gap between our dynamic testing problem formulation and the above static testing problem by showing that if $q_1=O(q_2)$, $p_{\init}\leq 0.5$ and if $q_1\leq \frac{1 - 1/\sqrt{2}}{\communitySize}$ and $q_2\leq \frac{1 - 1/\sqrt{2}}{\numitems}$, then the above two conditions on the prior vector are satisfied everyday in the discrete-time SIR stochastic block model parameterized by $(N,C,p_{\init},q_1,q_2)$. As a result the existing group testing strategies discussed above are order-optimal even for the dynamic testing problem formulation considered, provided that we use a sufficient number of tests each day to identify all new infections for that day. 
\end{itemize}


\subsection{Results on static group testing with non i.i.d priors}
We first consider the problem of static group testing, in which a person is infected independently with a known prior probability $p_i$. Denote by $\priors = (p_1,p_2,...,p_\numitems)$ the prior vector which collects the prior probabilities of infection of all individuals. We first define some  notation specific to this subsection:
\begin{itemize}[leftmargin=3mm]
	\item $\testmatrix$: test matrix
	\item $\defectives$: set of defectives or infections
	\item $\infectionstatuses=(\infectionstatus_1,\infectionstatus_2,...,\infectionstatus_\numitems)$: infection status configuration, i.e., individual $i$ is infected if and only if $\infectionstatus_i=1$.
	\item $\infectionstatuses(\defectives)\triangleq (\infectionstatus_1,...,\infectionstatus_\numitems)$ where $\infectionstatus_i=1$ iff $i\in \defectives$. Basically represents the vector notation for the set of infections given by $\defectives$. Note that there is a one-one correspondence between $\defectives$ and $\infectionstatuses(\defectives)$. We will use these two notations interchangeably based on convenience.
	\item $\testmatrix(\infectionstatuses)$ represents the test results corresponding to the given test design and infection status configuration. 
	\item For a fixed number of tests $\numtests$, define a decoding function $R:\{0,1\}^{[\numtests]}\rightarrow \{0,1\}^{[\numitems]}$ which estimates the infection statuses from the test results.
	\item A defective set $\defectives$ ``explains'' test results $\testresults$ iff $\testmatrix(\infectionstatuses(\defectives))=\testresults$.
	\item $\Pr(\infectionstatuses;\priors)$ denotes the probability of the infection status configuration under priors $\priors$, i.e. 
	$$\Pr(\infectionstatuses;\priors)=\prod_{i=1}^{\numitems} \prior_i^{\infectionstatus_i}(1-\prior_i)^{1-\infectionstatus_i}.$$
	\item Probability of error for a test matrix, decoder pair under given priors \begin{align*}\perror(\testmatrix, R;\priors)
		&\triangleq \expect_{\infectionstatuses\sim \priors} \mathbbm{1}\{R(\testmatrix(\infectionstatuses))\neq \infectionstatuses\} \\
		&= \sum_{\mathbf u \in \{0,1\}^N} \Pr(\infectionstatuses=\mathbf u;\priors) \mathbbm{1}\{R(\testmatrix(\mathbf u))\neq \mathbf u\}.
	\end{align*}
\end{itemize}

\begin{definition}[MAP decoder]
	 For fixed priors $\priors$ and testing matrix $\testmatrix$ with number of tests $\numtests$, we define the corresponding MAP decoder as $\mapdecoder(\ \cdot\ ;\testmatrix,\priors):\{0,1\}^{\numtests}\rightarrow \{0,1\}^{\numitems}$, where
	$$\mapdecoder(\testresults;\testmatrix,\priors)=\argmax\limits_{\infectionstatuses:\testmatrix(\infectionstatuses)=\testresults} \Pr(\infectionstatuses;\priors).$$
	In case of ties, the MAP decoder will select the solution which comes the earliest lexicographically.
\end{definition}

In words, the MAP decoder chooses the most likely configuration which explains the test results. We next show that the MAP decoder is the optimal decoder for a fixed $\testmatrix$ and $\priors$, i.e., the MAP decoder minimizes the probability of error amongst all decoders for any $\testmatrix$, $\priors$. 

 \textbf{Remark.} Though the MAP decoder is optimal, it is unclear if the optimization problem corresponding to the MAP decoder can be solved efficiently. However, many heuristics such as belief propagation (see for example \cite{GroupTesting-community-nonOverlap}) and random sampling methods exist which approximate well the MAP decoder. That said, in this work we use the MAP decoder only as a tool for theoretical analysis of the error probability.

\begin{restatable}[Optimality of MAP decoder]{lemma}{LemmaMAPDecoder}
\label{lem:MAP_decoder_optimality}
	For given test matrix $\testmatrix$ and priors $\priors$, the corresponding MAP decoder minimizes the probability of error for the test matrix under the given priors, i.e.,
	$$\perror(\testmatrix,\mapdecoder(\ \cdot\ ;\testmatrix,\priors);\priors)\leq \perror(\testmatrix,R;\priors)\ \forall R.$$
\end{restatable}

 We give the proof of \Cref{lem:MAP_decoder_optimality} to Appendix~\ref{app:MAP_decoder_opt_proof}.
 
 Given the optimality of the MAP decoder, we will denote by $\perror^*(G,\priors)\triangleq \perror(\testmatrix,\mapdecoder(\ \cdot\ ;\testmatrix,\priors);\priors),$  the optimal probability of error corresponding to a given test design and priors.

We next prove a property of the MAP decoder, and this property will be used in the proof of our main result that follows. The following Lemma  says that it is  easier for the MAP decoder to identify a sparser defective set. 

\begin{restatable}{lemma}{LemmaMAPSparser}
	\label{lem:map_decoder_increasing}
		\label{lemma:MAP_sparser}
		Consider a test matrix $\testmatrix$ and priors $\priors$. Suppose the corresponding MAP decoder is erroneous when identifying the defective set $\defectives$. Then the MAP decoder is also erroneous for the set of defectives is $\defectives\cup \{j\}$ with $p_j\leq 0.5$, i.e.,
		\begin{align*}
			\mathbbm{1}&\left\{\mapdecoder\left(\testmatrix(\infectionstatuses(\defectives\cup \{j\}));\testmatrix,\priors\right) \neq \infectionstatuses(\defectives\cup \{j\})\right\} \\
			& \geq \mathbbm{1}\left\{\mapdecoder(\testmatrix(\infectionstatuses(\defectives));\testmatrix,\priors) \neq \infectionstatuses(\defectives)\right\}.
		\end{align*}
\end{restatable}

 For the proof of \Cref{lem:map_decoder_increasing}, we refer the reader to Appendix~\ref{app:MAP_decoder_inc_proof}. 

We next prove the main new result for the static case. In words, the following theorem says that the group testing problem is only easier when the infections are sparser. As a result, this allows us to lower/upper bound the group testing problem with non-identical priors by a group testing problem with identical priors.

\begin{theorem}
	\label{thm:static_reduced}
	Consider a testing matrix $\testmatrix$ used with two different sets of priors $\priors$ and $\priors'$. Further let $\prior'_i=\prior_i$ for every $i\in [\numitems], i \neq j$ and $\prior'_{j}\leq \prior_j \leq 0.5$. The two prior vectors are same everywhere except at index $j$ where $\priors'$ is smaller. Then
	$$\perror^*(G,\priors') \leq \perror^*(G,\priors).$$
\end{theorem}

\begin{proof}
	We prove this by showing that when the MAP decoder corresponding to $(\testmatrix, \priors)$ pair is used as a decoder with $(\testmatrix, \priors')$, the probability of error is always lower, i.e.,
	\begin{align*}
		\perror(\testmatrix,&\mapdecoder(\ \cdot\ ;\testmatrix,\priors);\priors') \\
		&\leq \perror(\testmatrix,\mapdecoder(\ \cdot\ ;\testmatrix,\priors);\priors) =  \perror^*(G,\priors).
	\end{align*}
	As a result the optimal decoder for $(\testmatrix,\priors')$ pair has a probability of error not exceeding this quantity.

    Now, we can express the probability of error for the MAP decoder of $(\testmatrix,\priors)$ pair. For simplicity of notation in the following derivations, $\errorMapEvent(\defectives) \triangleq \mathbbm 1\left\{\mapdecoder(\testmatrix(\infectionstatuses(\defectives));\testmatrix,\priors) \neq \infectionstatuses(\defectives)\right\}$ denotes the indicator of the event that the MAP decoder is erroneous when the defective set is $\defectives$ (and under further assumptions that the priors are $\priors$ and test matrix is $\testmatrix$).
	\begin{align*}
		&\perror(\testmatrix,\mapdecoder(\ \cdot\ ;\testmatrix,\priors);\priors)\\ &= \sum_{\defectives \in [\numitems]} \Pr(\infectionstatuses(\defectives)) \errorMapEvent(\defectives)\\
		&= \sum_{\defectives \in [\numitems]} \prod_{i\in \defectives} p_i \prod_{l\in [\numitems]\setminus \defectives} (1-p_l) \errorMapEvent(\defectives)\\
		&\overset{(a)}{=} \sum_{\substack{\defectives \in [\numitems]\\|j\in \defectives}} \prod_{i\in \defectives} p_i \prod_{l\in [\numitems]\setminus \defectives} (1-p_l) \errorMapEvent(\defectives)\\
		&\hspace{0.3cm}+ \sum_{\substack{\defectives \in [\numitems]\\|j\notin \defectives}} \prod_{i\in \defectives} p_i \prod_{l\in [\numitems]\setminus \defectives} (1-p_l) \errorMapEvent(\defectives)\\
		&\overset{(b)}{=} p_j \sum_{\substack{\defectives \in [\numitems]\setminus\{j\}}} \prod_{i\in \defectives} p_i \prod_{l\in [\numitems]\setminus \defectives\cup\{j\}} (1-p_l) \errorMapEvent(\defectives\cup \{j\})\\
		&\hspace{0.3cm}+ (1-p_j)\sum_{\substack{\defectives \in [\numitems]\setminus\{j\}}} \prod_{i\in \defectives} p_i \prod_{l\in [\numitems]\setminus \defectives\cup\{j\}} (1-p_l) \errorMapEvent(\defectives) \numberthis,
	\end{align*}
	where in $(a)$ we split the summation into two cases -- one where $j\in \defectives$ and the other where $j\notin \defectives$; in $(b)$ we take $j$ out of the summation.
	
	Similarly, one could express the probability of error for the same decoder with the pair $(\testmatrix,\priors')$ as
\begin{align*}
	&\perror(\testmatrix,\mapdecoder(\ \cdot\ ;\testmatrix,\priors);\priors')\\
	&= p'_j \sum_{\substack{\defectives \in [\numitems]\setminus\{j\}}} \prod_{i\in \defectives} p_i \prod_{l\in [\numitems]\setminus \defectives\cup\{j\}} (1-p_l) \errorMapEvent(\defectives\cup \{j\})\\\
	&\hspace{0.1cm}+ (1-p'_j)\sum_{\substack{\defectives \in [\numitems]\setminus\{j\}}} \prod_{i\in \defectives} p_i \prod_{l\in [\numitems]\setminus \defectives\cup\{j\}} (1-p_l) \errorMapEvent(\defectives) \numberthis.
\end{align*}

	The first error term
	$\perror(\testmatrix,\mapdecoder(\ \cdot\ ;\testmatrix,\priors);\priors)$ is of the form $p_ja+(1-p_j)b$, and the second error term $\perror(\testmatrix,\mapdecoder(\ \cdot\ ;\testmatrix,\priors);\priors')$ is of the form $p'_ja+(1-p'_j)b$. From Lemma~\ref{lemma:MAP_sparser}, we have $\errorMapEvent(\defectives\cup \{j\}) \geq \errorMapEvent(\defectives)$ and hence $a\geq b$. Since $a\geq b$ and $p_j'\leq p_j$, one can verify that
	$p_ja+(1-p_j)b\geq p'_ja+(1-p'_j)b$, and thus $$\perror(\testmatrix,\mapdecoder(\ \cdot\ ;\testmatrix,\priors);\priors) \geq \perror(\testmatrix,\mapdecoder(\ \cdot\ ;\testmatrix,\priors);\priors'),$$ concluding the proof.
\end{proof}

Now one could repeatedly apply Theorem~\ref{thm:static_reduced} on the prior vector $\priors$ to conclude that any test matrix $\testmatrix$ should only do better on the reduced uniform prior vector $\priors_{\min}=(\prior_{\min},\prior_{\min},...,\prior_{\min})$ where $\prior_{\min}\triangleq \min_{i\in [N]}\prior_i$. On the other hand, the test matrix $\testmatrix$ should only do worse on the prior vector $\priors_{\max}=(\prior_{\max},\prior_{\max},...,\prior_{\max})$ where $\prior_{\max}\triangleq \max_{i\in [N]}\prior_i$. This is stated below without a formal proof.
\begin{corollary}
\label{cor:p_error}
	Consider a test matrix $\testmatrix$ and a prior vector $\priors$ such that $\prior_i \leq 0.5$ for all $i\in [\numtests]$. Let $\priors_{\min}=(\prior_{\min},\prior_{\min},...,\prior_{\min})$ where $\prior_{\min}\triangleq \min_{i\in [N]}\prior_i$ and let $\priors_{\max}=(\prior_{\max},\prior_{\max},...,\prior_{\max})$ where $\prior_{\max}\triangleq \max_{i\in [N]}\prior_i$. Then
	$$\perror^*(G,\priors_{\max}) \geq \perror^*(G,\priors)\geq \perror^*(G,\priors_{\min}).$$
\end{corollary}
As a consequence of the above corollary, the number of tests required to attain a fixed (small) probability of error $\epsilon$ with prior vector $\priors_{\min}$ is not more than the number of tests required to attain probability of error $\epsilon$ with prior vector $\priors$. This observation allows us to use the lower bound on the number of tests when the priors are identical. This is made precise in the following theorem.

\begin{theorem}
\label{thm:static_lower_bound}
    Consider the non-adaptive group testing problem with $\numitems$ items where  the probability of item $i$ being infected is $\prior_i \leq 0.5$. Let $\prior_{\min}\triangleq \min\limits_{i\in [\numitems]} \prior_i$. In order to achieve a probability of error  $\rightarrow 0$  as $N\rightarrow \infty$, the number of tests must be $$T(\priors) = \Omega(\min\{\numitems,\numitems \prior_{\min}\log \numitems\}).$$
\end{theorem}
\begin{proof}
    From Corollary \ref{cor:p_error}, suppose a test matrix $\testmatrix$ achieves a probability of error $\epsilon$ on prior vector $\priors$, the same test matrix achieves a probability of error not more than $\epsilon$ on the prior vector $\priors_{\min}$, where $\priors_{\min}=(\prior_{\min},\prior_{\min},...,\prior_{\min})$ and $\prior_{\min}\triangleq \min_{i\in [N]}\prior_i$. Any strategy that achieves a probability of error $\rightarrow 0$ as $N\rightarrow \infty$ with the prior vector $\priors_{\min}$ requires a number of tests equal to $\Omega(\min\{\numitems,\numitems \prior_{\min}\log \numitems\})$. Thus, we need at least as many tests with the prior vector $\priors$.
\end{proof}

 As discussed in Section~\ref{sec:tradGP}, the entropy bound in \Cref{lemma:entropy_LB} is an alternate lower bound on the number of tests needed for this problem. We note that the entropy bound might be greater or smaller than  the term $\numitems \prior_{\min}\log \numitems$  in Theorem~\ref{thm:static_lower_bound}. In particular, if $\prior_i\leq 1/2\ \forall i$ it is easy to see that $\sum_{i=1}^\numitems \binEntropy(p_i)\geq N \binEntropy(p_{\min}) \geq Np_{\min} \log \sfrac{1}{p_{\min}}$.  However the term $\sfrac{1}{p_{\min}}$ may be smaller or larger than $N$; thus our bound, that applies independently of the value of $p_{\min}$ (as long as $p_i\leq 0.5$) cannot be directly derived from the entropy bound, and could be either greater or lesser than the entropy bound.
  Having said that, the main advantage of the lower bound in Theorem~\ref{thm:static_lower_bound} is its particular form, which allows the proof of order-optimality of several static group testing algorithms, as we will see in the next subsection.


Now, if the prior vector $\priors$ is ``bounded'', in the sense that the maximum entry and minimum entry in $\priors$ differ by a constant factor (constant with respect to $\numitems$), then the  lower bound can be re-written in terms of the maximum entry in $\priors$ or the mean of $\priors$. Basically we here just use the fact that constant factors do not affect the order. We next make this corollary precise.

\begin{definition}[Bounded priors]
    Let $\eta \in [1,\infty)$ be a fixed constant (constant with respect to $\numitems$). A prior vector $\priors$ of length $\numitems$ is called $\eta-$bounded if
    $$\frac{\max_{i} p_i}{\min_i p_i} \leq \eta.$$
\end{definition}

\begin{corollary}[Lower bound for bounded priors]
\label{cor:order_bounded_priors}
    Consider the non-adaptive group testing problem with $\numitems$ items where  the probability of item $i$ being infected is $\prior_i \leq 0.5$. Let $\prior_{\max}\triangleq \max\limits_{i\in [\numitems]} \prior_i$ and $\prior_{\mean}\triangleq \frac{1}{\numitems}\sum_{i=1}^\numitems \prior_i$. Suppose $\priors=(\prior_1,...,\prior_\numitems)$ is $\eta$-bounded for some constant $\eta$. Any strategy that achieves a probability of error $\rightarrow 0$ as $N\rightarrow \infty$  requires
    \begin{align*}
   T(\priors) &= \Omega(\min\{\numitems,\numitems \prior_{\mean}\log \numitems\}) \\
	&=\Omega(\min\{\numitems,\numitems \prior_{\max}\log \numitems\}).
\end{align*}
\end{corollary}

\subsection{Performance of existing non-adaptive algorithms in the static non-identical priors}
\label{subsec:upper_bounds}
 Suppose $\priors$ is $\eta$-bounded and each $\prior_i\leq 0.5$. The following non-adaptive algorithms can be proved to be order-optimal with respect to the lower bound in \Cref{cor:order_bounded_priors}:
\begin{itemize}[leftmargin = 3mm]
    \item The Coupon Collector Algorithm (CCA) from \cite{prior} for prior vector $\priors$, as discussed in Section~\ref{sec:tradGP}, achieves a probability of error less than $2\numitems^{-\delta}$ with a number of tests less than $4e(1+\delta)\numitems \prior_{\mean} \log \numitems$ (see Theorem 3 in \cite{prior}). As a result, w.r.t to the lower bound in Corollary~\ref{cor:order_bounded_priors}, either CCA is order-optimal (if $N \geq \numitems \prior_{\mean}\log \numitems$)  or individual testing is order optimal (if $N \leq \numitems \prior_{\mean}\log \numitems$).
    \item  As discussed in Section~\ref{sec:tradGP} for the group testing problem with identical priors (say every item is infected with the same probability $\prior'$), a variety of randomized and explicit algorithms have been proposed\footnote{Most of these were considered in the context of combinatorial priors. However, Theorem 1.7 and Theorem 1.8 from \cite{GroupTestingMonograph} imply that any algorithm that attains a vanishing probability of error on the combinatorial priors, also attains a vanishing probability of error on the corresponding i.i.d probabilistic priors.} which achieve a vanishing probability of error with a number of tests $O(\numitems p'\log \numitems)$. From Corollary~\ref{cor:p_error}, any test matrix that achieves a vanishing probability of error with $\priors_{\max}$ should also attain a vanishing probability of error with $\priors$, and as a result $O(\numitems \prior_{\max}\log \numitems)$ tests are sufficient for the prior vector $\priors$. Consequently w.r.t our lower bound in Corollary~\ref{cor:order_bounded_priors}, any of these designs is order optimal (if $N \geq \numitems \prior_{\max}\log \numitems$)  or individual testing is order optimal (if $N \leq \numitems \prior_{\max}\log \numitems$).
\end{itemize}

\subsection{Dynamic testing - bridging the gap}
\label{subsec:dynamic_testing_results}
Given the discussion above, we next show conditions under which the prior probabilities of infections each day (these change everyday) are $\eta$-bounded and are each not more than $0.5$. If these two conditions are satisfied everyday for our discrete-time SIR stochastic block model set-up in Fig.~\ref{fig:problem_form}, then CCA and the other algorithms discussed in Section~\ref{subsec:upper_bounds} are order-optimal for our dynamic testing problem formulation. (see \ref{fig:problem_form}). We first define some notation, building upon the notation in Section~\ref{subsec:dynamic_testing_problem}.
\begin{itemize}[leftmargin=3mm]
	\item $p_{\max}^{(t)}\triangleq \max_{j}p_{j}^{(t)}$, the maximum probability of new infection on day $t$.
	\item $p_{\min}^{(t)}\triangleq \min_{j}p_{j}^{(t)}$, the minimum probability of new infection on day $t$.
\end{itemize}

 \begin{theorem}\label{thm:bounded_proof}
 	Consider the testing-intervention problem in Fig.~\ref{fig:problem_form} where the infections follow the discrete-time SIR stochastic block model $(\numitems,C, q_1,q_2,p_{\init})$. 
 	\begin{itemize}
 		\item[(i)] Suppose $p_{\init}\leq 0.5$, $q_1\leq \frac{1 - 1/\sqrt{2}}{\communitySize}$ and $q_2\leq \frac{1 - 1/\sqrt{2}}{\numitems}$, then $p_j^{(t)}\leq 0.5$. 
 		\item[(ii)] Suppose $\frac{q_1}{q_2}\leq \eta$, then $\frac{p_{\max}^{(t)}}{p_{\min}^{(t)}} \leq \eta$ and as a result the prior vector for each day is $\eta$-bounded.
 	\end{itemize}
 \end{theorem}
 \begin{proof}
 	We prove (i) first. We first have $p_j^{(0)} = p_{\init}\leq 0.5$. For $t \geq 1$, we have
 	\begin{align*}
 		p_j^{(t)}&=1-(1-q_1)^{I_j^{(t-1)}}(1-q_2)^{\sum_{j'\neq j}I^{(t-1)}_{j'}}\\
 		&\overset{(a)}{\leq} 1-(1-q_1)^{\communitySize}(1-q_2)^{\numitems}\\
 		&\overset{(b)}{\leq} 1-(1-\communitySize q_1)(1-\numitems q_2) \overset{(c)}{\leq} 0.5,
 	\end{align*}
 	where in $(a)$ we used the fact that the total number of infections inside a community and overall cannot be greater than $\communitySize$ and $\numitems$, respectively; $(b)$ follows because of the algebraic inequality $(1+x)^y \geq 1+xy$ if $x\geq -1$ and $y\notin (0,1)$; in $(c)$ we used our assumptions about $q_1$ and $q_2$.

    We next prove (ii). Since $q_1\geq q_2$ in our model, we have
	\begin{align*}
	    p_{j}^{(t)}&=1-(1-q_1)^{I_{j}^{(t-1)}}(1-q_2)^{\sum_{j'\neq j}I^{(t-1)}_{j'}}\\
	    &= 1-\left(\frac{1-q_1}{1-q_2}\right)^{I_{j}^{(t-1)}}(1-q_2)^{\sum_{j'}I^{(t-1)}_{j'}}\\
	    &\leq  1-\left(\frac{1-q_1}{1-q_2}\right)^{\max_j I^{(t-1)}_j}(1-q_2)^{\sum_{j'}I^{(t-1)}_{j'}},\numberthis
	    \label{eq:bounded_proof1_max}
	\end{align*}
 where $\max_j I^{(t)}_j$ is simply the maximum number of infections over all communities. 
    Likewise, 
    \begin{align*}
	    p_{j}^{(t)}&=1-(1-q_1)^{I_{j}^{(t-1)}}(1-q_2)^{\sum_{j'\neq j}I^{(t-1)}_{j'}}\\
	    &\overset{(a)}{\geq}  1-(1-q_2)^{\sum_{j'}I^{(t-1)}_{j'}},\numberthis
	    \label{eq:bounded_proof1_min}
	\end{align*}
	where in $(a)$ we used $q_2\leq q_1$.
Combining \eqref{eq:bounded_proof1_max} and~\eqref{eq:bounded_proof1_min} we have
\begin{align*}	
	\frac{p_{\max}^{(t)}}{p_{\min}^{(t)}} 
	&= \frac{1 -\left(\frac{1-q_1}{1-q_2}\right)^{\max_j I^{(t-1)}_j}(1-q_2)^{\sum_{j'}I^{(t-1)}_{j'}}} {1-(1-q_2)^{\sum_{j'}I^{(t-1)}_{j'}}} \\
	& \overset{(a)}{\leq} \frac{1-\left(\frac{1-q_1}{1-q_2}\right)^{\max_j I^{(t-1)}_j}(1-q_2)^{\max_j I^{(t-1)}_j}}{1-(1-q_2)^{\max_j I^{(t-1)}_j}}\\
	& =
	\frac{1-(1-q_1)^{\max_j I^{(t-1)}_j}}{1-(1-q_2)^{\max_j I^{(t-1)}_j}} \overset{(b)}{\leq}
	\frac{q_1}{q_2} \leq \eta.
\end{align*}
where $(a)$ follows from the following facts: 
the function $f_1(x) = \frac{1-\kappa x}{1-x}$ is increasing for $\kappa \in (0,1)$, 
the function $f_2(x) = (1 - q_2)^x$ is decreasing for $q_2 \in (0,1)$, 
and the sum $\sum_{j'}I^{(t-1)}_{j'}$ is lower bounded by $\max_j I^{(t-1)}_j$;
and $(b)$ follows from the fact that the function $f_3(x) =  \frac{1-(1-q_1)^{x}}{1-(1-q_2)^{x}}$ is decreasing in $x \geq 1$ for $q_1 \geq q_2$, and therefore the maximum of the ratio is obtained for $\max_j I^{(t-1)}_j = 1$.  
All proofs of the above statements are provided in Appendix~\ref{app:aux}.
\end{proof}

Finally, we make three remarks related to the results introduced in this section.\\
\textbf{Remark 1.} Both assumptions (i) and (ii) on the parameters in Theorem~\ref{thm:bounded_proof} will hold true when the number of communities is a constant, i.e., the size of each community is $C=\Theta(N)$ (as is the case when the population is well-mixed, or if we just consider a single community); assumption (i) does not require $C=\Theta(N)$. In our simulations, we  observed empirically that  assumption (ii) also holds when $C<<N$; we do not have a formal proof of Theorem~\ref{thm:bounded_proof} for this case however. \\
\textbf{Remark 2.} Our results hold not just for the specific model introduced in \Cref{sec:setup} (where in particular we assume symmetric intra and inter community transmissions) but for any underlying community structure where the two conditions (bounded prior vectors and the value of each prior not exceeding \sfrac{1}{2}) are satisfied. For example, one could have a model where an infected individual can transmit the infection only to a subset of his fellow community members with probability $q_1$ (he/she cannot transmit to the rest of the individuals in his/her community) and only to a subset of individuals outside his/her community with probability $q_2$. For this example model, the conditions in Theorem~\ref{thm:bounded_proof} are sufficient to prove the two requirements on the prior vector.\\
\textbf{Remark 3.} Intervention is a crucial aspect for our results to hold true. Without intervention in our dynamic model, many of the prior probabilities would be greater than $\sfrac{1}{2}$ and our requirements on the prior vector would not be satisfied.

%% file: evaluation.tex
\section{Numerical results} \label{sec:numerical}
In this section, we show illustrative numerical results on the necessary and sufficient number of tests required for the discrete-time SIR stochastic block model. We next describe the experimental set-up.
\begin{itemize}[leftmargin = 3mm]
    \item We  simulate multiple instances (or \textit{trajectories}) of the pipeline in Fig.~\ref{fig:TI_chain} where the infections follow the discrete-time SIR stochastic block model $(N,C,p_{\init},q_1,q_2)$, and for different testing strategies. We simulate 200 trajectories and in Fig.~\ref{fig:numerics}, plot the daily average of the quantities across these trajectories.
    \item For each of these testing strategies, we empirically find the number of tests required on each day to identify all the infections on the previous day. To do this, on each day for a given trajectory, we start with 1000 tests and decrease this number (at a certain granularity) until the testing strategy makes a mistake. The smallest number of tests for which the strategy worked is plotted.
    \item On the other hand, we also plot the \textit{entropy lower bound} in \Cref{lemma:entropy_LB}; it is easy to estimate this for our model via Monte-Carlo approximations. This bound holds for any set of values for $p_i^{(t-1)}$, regardless of whether the conditions required for Theorem~\ref{thm:bounded_proof} hold or not. The reason we use the entropy bound instead of our lower bound in Theorem~\ref{thm:static_lower_bound} is that the entropy bound was numerically observed to be larger. Indeed, the lower bound in Theorem~\ref{thm:static_lower_bound} contains some accompanying hidden constants which are small when used for our particular choice of $N$. 
\end{itemize}

\begin{figure}[t]
\centering
\captionsetup[subfigure]{margin=10pt}
\subcaptionbox{$(N,C,p_{\init},q_1,q_2) = (1000,20,0.02,0.03,0.0004)$.\label{fig:numerics1} \vspace{0.5cm}}
{\input{Figures/Numerics.tex}}
\subcaptionbox{$(N,C,p_{\init},q_1,q_2) = (1000,50,0.02,0.012,0.0004)$.\label{fig:numerics2}}
{\input{Figures/Numerics2.tex}}
\caption{\small Experimental results. We plot the average number of tests required by each strategy to identify the infection statuses of all non-isolated individuals each day for 2 different sets of parameters.}
\label{fig:numerics}
\end{figure}

We compare the following testing strategies in our numerical simulations.
\begin{itemize}[leftmargin =3mm]
    \item \textbf{Complete testing}. We test every non-isolated individual remaining in the population each day.
    \item \textbf{Coupon Collector Algorithm (CCA) from \cite{prior}.} We showed the order-optimality of this algorithm for the dynamic testing problem at the beginning of Section~\ref{subsec:dynamic_testing_results}. In short, on each day, the CCA algorithm constructs a random non-adaptive test design which depends on $p_j^{(t)}$. The idea is to place objects which are less likely to be infected in more number of tests and vice-versa. We refer the reader to \cite{prior} for the exact description of the algorithm.
    \item \textbf{Random group testing for max probability (Rnd. Grp. max.)} Here we construct a randomized design assuming that each individual has a prior probability of infection $p_{\max}^{(t)}$. From Corollary~\ref{cor:p_error}, such a design must also work for the actual priors $p_j^{(t)}$. We construct a constant column-weight design (see e.g. \cite{ncc-Johnson}) where each individual is placed in $L= \lfloor \sfrac{T}{(Np_{\max}^{(t)}\log 2 )} \rfloor$ tests. Such a test design achieves a vanishing probability of error with $O(Np_{\max}^{(t)} \log N)$ tests (see for example \cite{ncc-Johnson} for a proof), and hence is order-optimal under the conditions in Theorem~\ref{thm:bounded_proof}.
    \item \textbf{Random group testing for mean probability (Rnd. Grp. mean)} Here we construct a randomized design assuming that each individual has a prior probability of infection $p_{\mean}^{(t)}$, where $p_{\mean}^{(t)}$ is defined as the mean prior probability of infection across all individuals. Unlike Rnd. Grp. max., there is no guarantee on how many tests are needed by such a design to identify the infection statuses of all individuals. However, the numerical results in Fig.~\ref{fig:numerics} show that such a design  requires fewer tests than CCA or Rnd. Grp. max. designs.
\end{itemize}

The numerical results in Fig.~\ref{fig:numerics} are illustrated for two different parameter values of the discrete-time SIR stochastic block model. In both cases, we see that Rnd. Grp. mean requires the least number of tests to identify the infection statuses of all non-isolated individuals. Moreover, the number of tests required by all three testing strategies considered is much less than the number required by complete testing. In fact the numerics in Fig.~\ref{fig:numerics} indicate that if we use a number of tests equal to $\sfrac{1}{5}$ of number of tests required for complete individual testing,  all these algorithms would achieve the same performance as complete individual testing, at least for the particular examples that we considered. 

{A natural follow-up question to ask is if  there is a systematic way to choose the number of tests that need to be administered, given the upper bounds discussed in \Cref{subsec:upper_bounds}. In Appendix~\ref{app:heuristics}, we discuss one such heuristic and show that it achieves close-to-complete-testing performance.}


\section{Conclusions and open questions}
In this work we proposed the problem of dynamic group testing which asks the question of how to continually test given that infections spread during the testing period. Our numerical results answer the question we started with -- in the dynamic testing problem formulation, given a day of testing delay, is it possible to achieve close to complete testing performance with significantly fewer number of tests? The answer is yes, and in this paper we not only showed numerical evidence supporting this fact, but also gave theoretical bounds on the optimal number of tests needed in order to achieve this.

 Although we gave theoretical bounds on the optimal number of tests needed each day, many open questions remain. In particular, it would be interesting to study the same problem when tests are noisy, or when we can use other models of group tests such as the one in~\cite{ac-dc}, or simply when one cannot perfectly learn the states of all individuals on a testing day. In addition, it would also be of interest to study/use other test designs, for example like the ones in \cite{price2020fast, bondorf2020sublinear,bernoulli_testing1,bernoulli_testing2,PhaseTrans-SODA16,coja-oghlan19,ncc-Johnson,Nonadaptive-1,Nonadaptive-2,Nonadaptive-3,Saffron}. Finally, it remains open to see how these results translate to the continuous-time SIR stochastic network model of \cite{mathOfEpidemicsOnNetworks}.

%% file: Figures/Numerics.tex
\begin{tikzpicture}

\definecolor{color0}{rgb}{0.12156862745098,0.466666666666667,0.705882352941177}
\definecolor{color1}{rgb}{1,0.498039215686275,0.0549019607843137}
\definecolor{color2}{rgb}{0.172549019607843,0.627450980392157,0.172549019607843}
\definecolor{color3}{rgb}{0.83921568627451,0.152941176470588,0.156862745098039}
\definecolor{color4}{rgb}{0.580392156862745,0.403921568627451,0.741176470588235}
\definecolor{color5}{rgb}{0.549019607843137,0.337254901960784,0.294117647058824}

\begin{axis}[
width = 0.9\columnwidth,
height = 0.6\columnwidth,
legend cell align={left},
legend style={fill opacity=0.8, draw opacity=1, text opacity=1, at={(0.7,0.5)}, anchor=center, draw=white!80!black},
tick align=outside,
tick pos=left,
ylabel={Number of tests needed},
x grid style={white!69.0196078431373!black},
xmajorgrids,
xmin=-2.45, xmax=51.45,
xminorgrids,
xtick style={color=black},
y grid style={white!69.0196078431373!black},
ymajorgrids,
ymin=-50, ymax=1050,
yminorgrids,
ytick style={color=black}
]

\addplot [very thick, color1]
table {%
0 1000
1 980.13
2 961.66
3 945.02
4 930.8
5 919.03
6 909.34
7 901.07
8 894.41
9 889.11
10 885.03
11 882.17
12 879.9
13 877.75
14 876.13
15 874.87
16 873.96
17 873.24
18 872.58
19 872.13
20 871.81
21 871.62
22 871.49
23 871.4
24 871.33
25 871.27
26 871.21
27 871.15
28 871.12
29 871.09
30 871.07
31 871.05
32 871.01
33 871
34 870.99
35 870.99
36 870.99
37 870.99
38 870.99
39 870.99
40 870.99
41 870.99
42 870.99
43 870.99
44 870.99
45 870.99
46 870.99
47 870.99
48 870.99
49 870.99
};
\addlegendentry{Complete}
\addplot [very thick, color2]
table {%
0 255
1 243.77
2 221.53
3 195.66
4 169.31
5 143.3
6 126.6
7 99.48
8 83.39
9 74.71
10 67.09
11 63.34
12 49.9
13 48.89
14 36.62
15 35.07
16 23.18
17 42.55
18 36.33
19 26.67
20 25.17
21 23.2
22 14.64
23 12.18
24 12.13
25 11.68
26 11.67
27 11.25
28 11.25
29 11.25
30 11.25
31 11.25
32 11.25
33 11.25
34 11.25
35 11.25
36 11.25
37 11.25
38 11.25
39 11.25
40 11.25
41 11.25
42 11.25
43 11.25
44 11.25
45 11.25
46 11.25
47 11.25
48 11.25
49 11.25
};
\addlegendentry{Rnd. Grp. mean}
\addplot [very thick, color3]
table {%
0 319
1 292.19
2 263.53
3 225.03
4 188.54
5 171.75
6 141
7 129.39
8 105.58
9 93.71
10 82.48
11 65.54
12 58.85
13 39.61
14 35.61
15 30.33
16 26.04
17 21.87
18 20.78
19 19.3
20 19.66
21 15.65
22 13.37
23 12.43
24 11.71
25 11.63
26 11.18
27 11.52
28 10.51
29 10.51
30 10.51
31 10.51
32 10.51
33 10.51
34 10.51
35 10.51
36 10.51
37 10.51
38 10.51
39 10.51
40 10.51
41 10.51
42 10.51
43 10.51
44 10.51
45 10.51
46 10.51
47 10.51
48 10.51
49 10.51
};
\addlegendentry{CCA}
\addplot [very thick, color4]
table {%
0 253
1 379.24
2 401.39
3 374.44
4 341.25
5 305.91
6 288.34
7 222.55
8 193.55
9 158.2
10 142.9
11 118.35
12 110.26
13 84.67
14 70.53
15 52.15
16 39.63
17 25.29
18 22.79
19 16.77
20 15.86
21 15.57
22 12.25
23 12.25
24 12.25
25 12.25
26 12.25
27 12.25
28 12.25
29 12.25
30 12.25
31 12.25
32 12.25
33 12.25
34 12.25
35 12.25
36 12.25
37 12.25
38 12.25
39 12.25
40 12.25
41 12.25
42 12.25
43 12.25
44 12.25
45 12.25
46 12.25
47 12.25
48 12.25
49 12.25
};
\addlegendentry{Rnd. Grp. max.}
\addplot [very thick, color5]
table {%
0 141.440542541821
1 122.462374972564
2 108.90236610893
3 95.9660005827493
4 82.1748197130668
5 68.6664428854234
6 57.0567712815792
7 48.4751991378954
8 39.7350801988969
9 32.0986044705848
10 25.1104483262534
11 17.7991796578127
12 14.31014550562
13 13.5919728503754
14 10.258799700635
15 7.91599053363005
16 5.79711564222216
17 4.67602209205896
18 4.2758349514134
19 2.85547358065635
20 1.99016637224024
21 1.17674811751419
22 0.827960695851604
23 0.566902209505912
24 0.407337331862889
25 0.344860648808978
26 0.361679039560578
27 0.366413631226672
28 0.196928093605589
29 0.19098005711792
30 0.140310131679218
31 0.145731039033794
32 0.24927616805268
33 0.0696914770706878
34 0.077212549007092
35 0
36 0
37 0
38 0
39 0
40 0
41 0
42 0
43 0
44 0
45 0
46 0
47 0
48 0
49 0
};
\addlegendentry{Lower bound}
\end{axis}

\end{tikzpicture}

%% file: Figures/Numerics2.tex
\begin{tikzpicture}

\definecolor{color0}{rgb}{0.12156862745098,0.466666666666667,0.705882352941177}
\definecolor{color1}{rgb}{1,0.498039215686275,0.0549019607843137}
\definecolor{color2}{rgb}{0.172549019607843,0.627450980392157,0.172549019607843}
\definecolor{color3}{rgb}{0.83921568627451,0.152941176470588,0.156862745098039}
\definecolor{color4}{rgb}{0.580392156862745,0.403921568627451,0.741176470588235}
\definecolor{color5}{rgb}{0.549019607843137,0.337254901960784,0.294117647058824}

\begin{axis}[
width = 0.9\columnwidth,
height = 0.6\columnwidth,
legend cell align={left},
legend style={fill opacity=0.8, draw opacity=1, text opacity=1, at={(0.7,0.5)}, anchor=center, draw=white!80!black},
tick align=outside,
tick pos=left,
ylabel={Number of tests needed},
x grid style={white!69.0196078431373!black},
xmajorgrids,
xmin=-2.45, xmax=51.45,
xminorgrids,
xtick style={color=black},
y grid style={white!69.0196078431373!black},
ymajorgrids,
ymin=-50, ymax=1050,
yminorgrids,
ytick style={color=black}
]

\addplot [very thick, color1]
table {%
0 1000
1 980.19
2 961.61
3 945.02
4 930.1
5 917.64
6 907.26
7 898.46
8 891.03
9 884.71
10 879.44
11 875.52
12 872.34
13 869.71
14 867.74
15 866.28
16 865.03
17 864.04
18 863.14
19 862.47
20 861.88
21 861.4
22 861.1
23 860.8
24 860.57
25 860.37
26 860.21
27 860.15
28 860.07
29 860.03
30 860
31 859.96
32 859.92
33 859.87
34 859.84
35 859.83
36 859.83
37 859.83
38 859.83
39 859.83
40 859.83
41 859.83
42 859.83
43 859.83
44 859.83
45 859.83
46 859.83
47 859.83
48 859.83
49 859.83
};
\addlegendentry{Complete}
\addplot [very thick, color2]
table {%
0 257
1 241.63
2 216.83
3 192.39
4 185.17
5 165.92
6 142.13
7 128.5
8 111.71
9 102.08
10 85.61
11 79.25
12 68.97
13 65.23
14 67.59
15 35.98
16 31.28
17 27.03
18 49.72
19 22.3
20 20.43
21 19.63
22 18.41
23 21.61
24 15.16
25 14.05
26 12.45
27 19.99
28 14.71
29 14.1
30 20.12
31 12.4
32 11.89
33 18.82
34 12.04
35 14.33
36 13.34
37 11.98
38 12.48
39 11.72
40 11.6
41 11.35
42 10.95
43 10.54
44 10.54
45 10.54
46 10.54
47 10.54
48 10.54
49 10.54
};
\addlegendentry{Rnd. Grp. mean}
\addplot [very thick, color3]
table {%
0 318
1 281.54
2 285.86
3 246.33
4 211.2
5 177.87
6 155.24
7 139.74
8 123.81
9 100.01
10 85.51
11 76.77
12 67.51
13 63.33
14 52.49
15 47.01
16 37.68
17 28.26
18 25.22
19 23.92
20 19.51
21 18.96
22 19.75
23 14.93
24 14.14
25 14.75
26 12.68
27 12.67
28 12.66
29 12.22
30 12.22
31 12.22
32 12.22
33 12.22
34 12.22
35 12.22
36 12.22
37 12.22
38 12.22
39 12.22
40 12.22
41 12.22
42 12.22
43 12.22
44 12.22
45 12.22
46 12.22
47 12.22
48 12.22
49 12.22
};
\addlegendentry{CCA}
\addplot [very thick, color4]
table {%
0 242
1 292.52
2 288.07
3 267.08
4 236.02
5 208.9
6 184.27
7 160.84
8 142.89
9 120.6
10 102.01
11 85.81
12 73.37
13 61.67
14 50.06
15 42.15
16 35.93
17 35.35
18 31.04
19 24.35
20 19
21 15.76
22 13.63
23 14.11
24 14.37
25 13.54
26 12.45
27 11.66
28 11.64
29 11.62
30 11.23
31 11.23
32 11.23
33 11.23
34 11.23
35 11.23
36 11.23
37 11.23
38 11.23
39 11.23
40 11.23
41 11.23
42 11.23
43 11.23
44 11.23
45 11.23
46 11.23
47 11.23
48 11.23
49 11.23
};
\addlegendentry{Rnd. Grp. max.}
\addplot [very thick, color5]
table {%
0 141.440542541821
1 128.496048708674
2 116.876779033311
3 103.703904463099
4 91.9233872932414
5 76.7680939710161
6 64.5568988065168
7 55.3427946363966
8 46.9858009481585
9 40.0828965048352
10 33.8619460664428
11 25.7621175783538
12 20.9177323685867
13 16.8371194476358
14 12.7535531411561
15 9.52296846222672
16 8.14813276954168
17 6.60805082943501
18 5.92017814446982
19 4.30669391566364
20 3.80643929577771
21 3.06327136539248
22 2.05492279919286
23 2.09688436305676
24 1.57109334063125
25 1.33314981519294
26 1.03489175813224
27 0.395201438818143
28 0.541676315687007
29 0.282185729737887
30 0.217655195291139
31 0.269227200413146
32 0.27633228460526
33 0.339566555916878
34 0.216538733884447
35 0.0752305250396269
36 0
37 0
38 0
39 0
40 0
41 0
42 0
43 0
44 0
45 0
46 0
47 0
48 0
49 0
};
\addlegendentry{Lower bound}
\end{axis}

\end{tikzpicture}

%% file: appendix.tex
\section{Comparison of discrete and continuous-time SIR models}
\label{app:discrete_vs_cont}
 The well-studied continuous-time SIR stochastic network model from \cite{mathOfEpidemicsOnNetworks} has been the main motivation for our discrete-time SIR stochastic block model. In fact, the discrete-time SIR stochastic block model described in \Cref{subsec:model} can be considered as a discretized version of the continuous-time SIR stochastic network model over the weighted graph, where 2 individuals belonging to the same community are connected by an edge with weight $q_1$ and 2 individuals belonging to different communities are connected by an edge with weight $q_2$, and recoveries occur at the rate $r$/day -- i.e., an infected individual transmits the disease to a susceptible individual in the same community at the rate $q_1$/day and to a susceptible individual in a different community at the rate $q_2$/day. In Fig.~\ref{fig:cont_vs_discrete}, we compare the continuous-time model above and the discrete-time model for a few example values of $q_1$, $q_2$ and $r$ for illustration.

\begin{figure}
    \centering
    \input{Figures/Discrete_vs_cont}
    \caption{\small Continuous vs discrete-time model. Continuous model in dashed and discrete model in solid curves. Recovery probability $r=0.1$ in all cases.}
    \label{fig:cont_vs_discrete}
\end{figure}
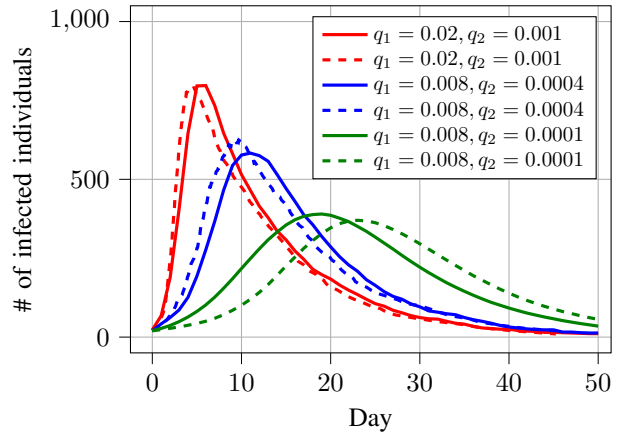

We make a few observations:\\
$\bullet$  The progression of the disease in the discrete-time and continuous-time models, though not identical, follow a similar pattern, justifying the use of the discrete-time model.\\
$\bullet$ In both the models, $\sfrac{1}{q_1}$ is the expected time for an infected individual to transmit the disease to a susceptible individual in the same community, $\sfrac{1}{q_2}$ is the expected time for an infected individual to transmit the disease to a susceptible individual in a different community and $\sfrac{1}{r}$ is the expected time for an infected individual to recover.\\
$\bullet$ In the continuous-time model, an individual can get infected and recovered in the same day, whereas this is not possible in our discrete-time model (infected individuals can recover starting from the day after they are infected).

\section{Proof of Lemma~\ref{lem:MAP_decoder_optimality}}
\label{app:MAP_decoder_opt_proof}
The optimality of the MAP decoder is a standard result in statistics and signal processing. We however give the proof in the context of our problem, for completeness.

\LemmaMAPDecoder*

\begin{proof}
As stated at the beginning of  \Cref{sec:results}, 
the Probability of error for a test matrix, decoder pair under given the priors is \begin{align*}\perror(\testmatrix, R;\priors)
	&\triangleq \expect_{\infectionstatuses\sim \priors} \mathbbm{1}\{R(\testmatrix(\infectionstatuses))\neq \infectionstatuses\} \\
	&= \expect_{\mathbf Y} \; \expect_{\infectionstatuses | \mathbf Y } \mathbbm{1}\{R(\testmatrix(\infectionstatuses)) \neq \infectionstatuses \},
\end{align*}
where $\mathbf Y$ is the set of test results.
For the MAP decoder, the term inside $\expect_{\mathbf Y}$ is
\begin{align*}
	\expect_{\infectionstatuses | \mathbf Y}&  \mathbbm{1}\{\mapdecoder(\mathbf Y ;\testmatrix,\priors))\neq \infectionstatuses \} \\
	&= \expect_{\infectionstatuses |\mathbf Y} \mathbbm{1}\{\argmax\limits_{\infectionstatuses:\testmatrix(\infectionstatuses)=\mathbf Y} \Pr(\infectionstatuses;\priors) \neq \infectionstatuses \} \\
	&= \sum_{ \defectives \in \left[\numitems\right] } 
	\Pr\left(\infectionstatuses(\defectives) | \mathbf Y ; \priors\right)\\  &\hspace{1.2cm} \cdot \mathbbm{1}\{\argmax\limits_{\infectionstatuses(\defectives):\testmatrix(\infectionstatuses(\defectives))=\mathbf Y} \Pr(\infectionstatuses(\defectives);\priors) \neq \infectionstatuses(\defectives)\}\\
	&=1 - \Pr\left(\argmax\limits_{\infectionstatuses:\testmatrix(\infectionstatuses)=\mathbf Y} \Pr(\infectionstatuses |\mathbf Y;\priors) \right)\\
	& = 1- \max_{\infectionstatuses:\testmatrix(\infectionstatuses)=\mathbf Y} \frac{\Pr (\infectionstatuses; \priors)}{\Pr(\mathbf Y)} \numberthis \label{eq:MAP_decoder_optimality:MAP}.
\end{align*}

Similarly, for any decoder $R$, we have
\begin{align*}
	\expect_{\infectionstatuses | \mathbf Y}&  \mathbbm{1}\{R(\mathbf Y) \neq \infectionstatuses \} \\
	&= \sum_{ \defectives \in \left[\numitems\right] } 
	\Pr\left(\infectionstatuses(\defectives) | \mathbf Y ; \priors\right) \cdot \mathbbm{1}\{R(\mathbf Y) \neq \infectionstatuses\}\\
	&=1 - \Pr\left(R(\mathbf Y) |  \mathbf Y;\priors\right) \geq 1- \max_{\infectionstatuses:\testmatrix(\infectionstatuses)=\mathbf Y} \frac{\Pr (\infectionstatuses; \priors)}{\Pr( \mathbf Y)} \numberthis \label{eq:MAP_decoder_optimality:R}.
\end{align*}
Comparing \eqref{eq:MAP_decoder_optimality:MAP} and \eqref{eq:MAP_decoder_optimality:R} concludes the proof.
\end{proof}

\section{Proof of Lemma~\ref{lem:map_decoder_increasing}}
\label{app:MAP_decoder_inc_proof}
\LemmaMAPSparser*
\begin{proof}
We first state the trivial case where $\mathbbm{1}\left\{\mapdecoder(\testmatrix(\infectionstatuses(\defectives));\testmatrix,\priors) \neq \infectionstatuses(\defectives)\right\} = 0$. 
Under that assumption, the inequality of \Cref{lem:map_decoder_increasing} always holds.

We then consider the case where $\mathbbm{1}\left\{\mapdecoder(\testmatrix(\infectionstatuses(\defectives));\testmatrix,\priors) \neq \infectionstatuses(\defectives)\right\} = 1$, i.e., the MAP decoder makes an error when the defective set is $\defectives$.
In that case, one of the two situations is possible:
\begin{itemize}[leftmargin=5mm]
    \item[(1)] there exists some set $\defectives' \neq \defectives$, such that $\testmatrix (\infectionstatuses (\defectives')) = \testmatrix (\infectionstatuses (\defectives))$ and
$	\Pr(\infectionstatuses (\defectives');\priors) > \Pr(\infectionstatuses (\defectives);\priors)$ or
    \item[(2)] there exists some set $\defectives' \neq \defectives$, such that $\testmatrix (\infectionstatuses (\defectives')) = \testmatrix (\infectionstatuses (\defectives))$ and
$	\Pr(\infectionstatuses (\defectives');\priors) = \Pr(\infectionstatuses (\defectives);\priors)$ and $\defectives'$ is lexicographically earlier than $\defectives$. 
\end{itemize}
Hence MAP identifies incorrectly $\defectives'$ as the defective set given that $\defectives$ was the true defective set. We prove assuming that the first situation occurred; the proof follows identical arguments for the second situation.

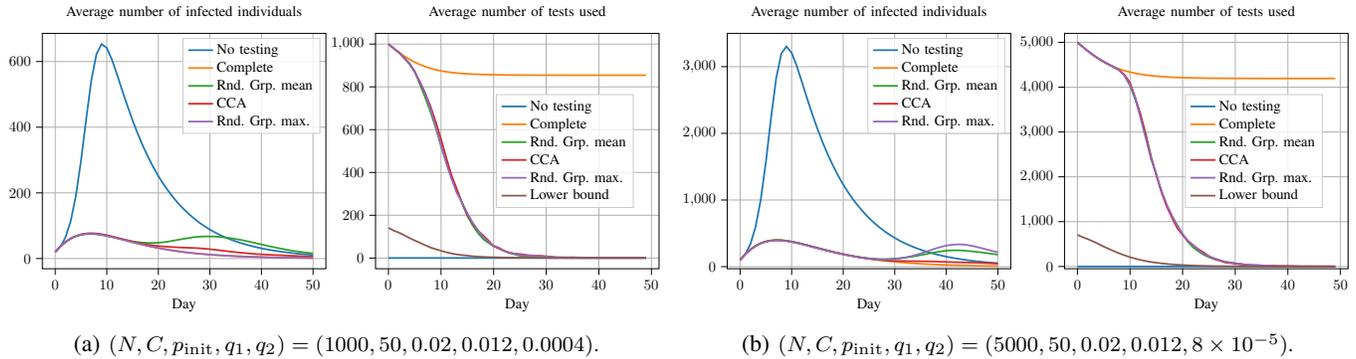
\begin{figure*}
\centering
\captionsetup[subfigure]{margin=10pt}
\subcaptionbox{\footnotesize $(N,C,p_{\init},q_1,q_2) = (1000,50,0.02,0.012,0.0004)$.\label{fig:heuristic1} }
{\scalebox{0.55}{\input{Figures/heuristic_1}}}
\subcaptionbox{\footnotesize $(N,C,p_{\init},q_1,q_2) = (5000,50,0.02,0.012,8\times 10^{-5})$.\label{fig:heuristic2}}
{\scalebox{0.55}{\input{Figures/heuristic_2}}}
\caption{\small  Experimental results for the heuristic procedure described in Appendix~\ref{app:heuristics}.  We plot the average number of infected individuals and the number of tests used as a function of time (in days). For comparison, we also plot the performance when no one is tested (no testing) and when everyone is tested (Complete).}
\label{fig:heuristics}
\end{figure*}
	
Now, we consider two different cases for individual $j$ that is added to $\defectives \cup \{j\}$:\\
(i) If $j \notin \defectives'$, then from our assumption in (1), notice that the defective set $\defectives' \cup \{j\}$ explains the test results of $\defectives \cup \{j\}$ -- $\defectives'$ gives the same test results as $\defectives$ and the extra individual $j$ added to both the sets will still give the same results. We next claim that $\Pr(\infectionstatuses (\defectives'\cup \{j\});\priors)> \Pr(\infectionstatuses (\defectives \cup \{j\});\priors),$ and consequently the MAP decoder will fail to correctly identify the defective set $\defectives' \cup \{j\}$. Now to prove our claim, we start with our assumption (1), i.e., 
\begin{align*}
    &\Pr(\infectionstatuses (\defectives');\priors) > \Pr(\infectionstatuses (\defectives);\priors)\\
    &\implies \prod_{i\in \defectives'}p_i \prod_{l\in [N]\setminus \defectives'} (1-p_l) > \prod_{i\in \defectives}p_i \prod_{l\in [N]\setminus \defectives} (1-p_l)\\
    &\overset{(a)}{\implies} (1-p_j)\prod_{i\in \defectives'}p_i \prod_{l\in [N]\setminus \defectives'\cup \{j\} } (1-p_l) \\ 
    &\hspace{2cm}> (1-p_j) \prod_{i\in \defectives}p_i \prod_{l\in [N]\setminus \defectives\cup \{j\} } (1-p_l)\\
    &\overset{(b)}{\implies} p_j\prod_{i\in \defectives'}p_i \prod_{l\in [N]\setminus \defectives'\cup \{j\} } (1-p_l) \\ 
    &\hspace{2cm}> p_j \prod_{i\in \defectives}p_i \prod_{l\in [N]\setminus \defectives\cup \{j\} } (1-p_l) \\
    &\overset{(c)}{\implies} \prod_{i\in \defectives'\cup \{j\}}p_i \prod_{l\in [N]\setminus \defectives'\cup \{j\} } (1-p_l) \\ 
    &\hspace{2cm}>  \prod_{i\in \defectives\cup \{j\}}p_i \prod_{l\in [N]\setminus \defectives\cup \{j\} } (1-p_l) \\
    &\implies \Pr(\infectionstatuses (\defectives'\cup \{j\});\priors) > \Pr(\infectionstatuses (\defectives \cup \{j\});\priors),
\end{align*}
where in $(a)$ we take out the term corresponding to $j$, also we use the fact that $j\notin \defectives$ and $j\notin \defectives'$; $(b)$ follows from multiplying both sides with $\sfrac{p_j}{1-p_j}$; in $(c)$ we push the $p_j$ term into the first product term.

(ii) If $j \in \defectives'$, we again first note that the defective set $\defectives'\cup \{j\}=\defectives'$ explains the test results of $\defectives \cup \{j\}$.  We next claim that $\Pr(\infectionstatuses (\defectives');\priors)> \Pr(\infectionstatuses (\defectives \cup \{j\});\priors),$ and consequently the MAP decoder will fail to correctly identify the defective set $\defectives' \cup \{j\}$. Now to prove our claim, we start with our assumption (1), i.e., 
\begin{align*}
    &\Pr(\infectionstatuses (\defectives');\priors) > \Pr(\infectionstatuses (\defectives);\priors)\\
    &\implies \prod_{i\in \defectives'}p_i \prod_{l\in [N]\setminus \defectives'} (1-p_l) > \prod_{i\in \defectives}p_i \prod_{l\in [N]\setminus \defectives} (1-p_l)\\
    &\overset{(a)}{\implies} p_j \prod_{i\in \defectives'\setminus\{j\}}p_i \prod_{l\in [N]\setminus \defectives' } (1-p_l) \\ 
    &\hspace{2cm}> (1-p_j) \prod_{i\in \defectives}p_i \prod_{l\in [N]\setminus \defectives\cup \{j\} } (1-p_l)\\
    &\overset{(b)}{\implies} p_j\prod_{i\in \defectives'\setminus\{j\}}p_i \prod_{l\in [N]\setminus \defectives' } (1-p_l) \\ 
    &\hspace{2cm}> p_j \prod_{i\in \defectives}p_i \prod_{l\in [N]\setminus \defectives\cup \{j\} } (1-p_l) \\
    &\overset{(c)}{\implies} \prod_{i\in \defectives'}p_i \prod_{l\in [N]\setminus \defectives' } (1-p_l) \\ 
    &\hspace{2cm}>  \prod_{i\in \defectives\cup \{j\}}p_i \prod_{l\in [N]\setminus \defectives\cup \{j\} } (1-p_l) \\
    &\implies \Pr(\infectionstatuses (\defectives');\priors) > \Pr(\infectionstatuses (\defectives \cup \{j\});\priors),
\end{align*}
where in $(a)$ we take out the term corresponding to $j$, also note that $j\notin \defectives$ but $j\in \defectives'$; $(b)$ follows from the fact that $1-p_j\geq p_j$ when $p_j\leq 0.5$, so we can replace the $(1-p_j)$ term on the right-hand side by $p_j$ without affecting the inequality; in $(c)$ we push the $p_j$ term into the first product term.
\end{proof}

\section{Auxiliary results for~\Cref{thm:bounded_proof}}
\label{app:aux}
In this section we prove some auxiliary statements about functions $f_1(x)$, $f_2(x)$ and $f_3(x)$ that are used at the end of the proof of \Cref{thm:bounded_proof}:\\
$\bullet$ $f_1(x) = \frac{1-\kappa x}{1-x}$ is increasing for $\kappa \in (0,1)$, because $f'_1(x) =  -\frac{\kappa-1}{\left(x-1\right)^2}>0$.\\
$\bullet$ $f_2(x) = (1 - q_2)^x$ is decreasing for $q_2 \in (0,1)$, because $f'_2(x)  = \ln\left(1-q_2\right)\left(1-q_2\right)^x<0$.\\
$\bullet$ $f_3(x) =  \frac{1-c_1^{x}}{1-c_2^{x}}$ is decreasing for $q_1 \geq q_2$, because of the following:
Let $c_1 = 1-q_1$ and $c_2 = 1-q_2$, so that $c_2 \geq c_1$. Then,
\begin{align*}
	f'_3(x)  
	&= \frac{1}{\left(1-c_2^x\right)^2}\left( (1-c_1)^x c_2^x \ln{c_2} - \left(1-c_2^x\right) c_1^x \ln{c_1} \right) \\
	&= \frac{1}{x \left(1-c_2^x\right)^2} \left( (1-c_1)^x c_2^x \ln{c_2^x} - \left(1-c_2^x\right) c_1^x \ln{c_1^x}\right) \\
	&= \frac{(1-c_1)^x (1-c_2)^x}{x \left(1-c_2^x\right)^2} \left( \frac{ c_2^x \ln{c_2^x}}{(1-c_2)^x} - \frac{ c_1^x \ln{c_1^x}}{(1-c_1)^x}\right) \overset{(a)}{\leq} 0,
\end{align*} 
where $(a)$ follows from the fact that $c_2 \leq c_1$ and the function $g(c) = \frac{c \ln{c}}{1-c}$ is non-increasing for $c \in (0,1)$.
The latter can be seen by taking the derivative $g'(c) = \frac{\ln{c}-c+1}{\left(1-c\right)^2}$, which is always non-positive for $c \in (0,1)$, as $\ln{c} \leq c - 1$.

\section{A heuristics for dynamic group testing}
\label{app:heuristics}
Given the results and discussion in \Cref{sec:numerical}, a natural question to ask is if one could use a number of tests based on the upper bounds discussed in \Cref{subsec:upper_bounds}. In particular, we focus on the upper bound for CCA which implies that CCA achieves a probability of error less than $2\numitems^{-\delta}$ with a number of tests at most $4e(1+\delta)\numitems \prior_{\mean} \log \numitems$ (see Theorem 3 in \cite{prior}). Note that the probability of error is small, but not zero, for finite values of $N$. Here, we use a number of tests equal to $12e \numitems \prior_{\mean} \log \numitems$ each day (corresponding to an error probability less than $2\numitems^{-2}$) and plot the number of errors made by each of the three test designs considered in \Cref{sec:numerical}. The experimental set-up is as follows: \\
$(i)$ We maintain an estimate of the probability $p_j^{(t)}$ that a susceptible individual belonging to community $j$ becomes infected on day $t$ (see \Cref{subsec:model} for the precise definition), for each community $j$, and for each day $t$.\\
$(ii)$ At the beginning of day $t$, we obtain the results of the tests administered on day $t-1$. From these results, we form an estimate $\widehat U_i^{(t-1)}$ of the infection statuses $U_i^{(t-1)}$  of individual $i$ at the beginning of day $t-1$, for each $i$. In order to learn the statuses, we use the \textit{ Definite Defective} (DD) decoder (see Section 2.4 in \cite{GroupTestingMonograph}) which is guaranteed to have no false positives. Indeed, one could use more sophisticated decoders, such as ones based on loopy belief propagation. However, these decoders potentially give rise to both false positives and false negatives, resulting in an unfair comparison across different algorithms\footnote{Indeed, this begs the very complicated comparison between the impact of false positives and false negatives, which we avoid for the sake of simplicity.}.\\
$(iii)$ We isolate all individuals $i$ where $\widehat U_i^{(t-1)}=1$.\\
$(iv)$ We update $p_j^{(t-1)}$ using our estimates $\widehat U_i^{(t-1)}$.\\
$(v)$ Using our estimates of $p_j^{(t-1)}$, we estimate the value of $p_{\mean}^{(t)}$ and choose a number of tests $$T=\min\{12e \numitems^{(t)} \prior_{\mean}^{(t)} \log \numitems^{(t)},\numitems^{(t)}\},$$ where $\numitems^{(t)}$ is the current number of non-isolated individuals in the community. We next construct a testing matrix with $T$ tests and administer these tests. For complete testing, we use $T=\numitems^{(t)}$.\\
$(vi)$ Steps $(ii)-(v)$ repeat each day.

Given the above set-up, \Cref{fig:heuristics} compares the performance of the test designs described in \Cref{sec:numerical}. We make a few observations:\\
$\bullet$ We see that the algorithms do not always attain the performance of complete testing. This is due to the fact that for finite $N$, the probability of error is non-zero. However, as seen from  \Cref{fig:heuristics}, the performance of CCA improves as $N$ increases. On the other hand, Rnd. Grp. max. has the opposite trend; this is not surprising since the number of tests was chosen based on the upper bound for CCA and as a result there is no guarantee that the same number of tests is sufficient for Rnd. Grp. max.\\
$\bullet$ In comparison to the plots in \cref{fig:numerics}, the number of tests used here is much higher during the initial few days, which indicates the looseness of the upper bound; it remains open to show tighter upper bounds for these algorithms.\\
$\bullet$ Suppose we make an error when identifying the infection status on a particular day $t$, the estimates of $p_j^{(t-1)}$ are not exact, which in turn leads to potentially insufficient choices for the number of tests needed for subsequent days and inaccurate test designs. This drives an error accumulation and as a result the later days are more prone to error, as also seen in \Cref{fig:heuristics}.

%% file: Figures/Discrete_vs_cont.tex
\begin{tikzpicture}

\definecolor{color0}{rgb}{0.12156862745098,0.466666666666667,0.705882352941177}
\definecolor{color1}{rgb}{1,0.498039215686275,0.0549019607843137}
\definecolor{color2}{rgb}{0.172549019607843,0.627450980392157,0.172549019607843}
\definecolor{color3}{rgb}{0.83921568627451,0.152941176470588,0.156862745098039}
\definecolor{color4}{rgb}{0.580392156862745,0.403921568627451,0.741176470588235}
\definecolor{color5}{rgb}{0.549019607843137,0.337254901960784,0.294117647058824}

\begin{axis}[
width = 0.9\columnwidth,
height = 0.7\columnwidth,
legend cell align={left},
legend pos=north east,
legend style = {nodes={scale=0.8, transform shape}},
tick align=outside,
tick pos=left,
x grid style={white!69.0196078431373!black},
xmajorgrids,
xmin=-2.45, xmax=51.45,
xminorgrids,
xlabel = Day,
ylabel = \# of infected individuals,
xtick style={color=black},
y grid style={white!69.0196078431373!black},
ymajorgrids,
ymin=-50, ymax=1050,
yminorgrids,
ytick style={color=black}
]
\addplot [very thick, red]
table {%
0 23
1 68
2 189
3 405
4 646
5 796
6 797
7 734
8 646
9 581
10 515
11 465
12 415
13 378
14 336
15 308
16 273
17 246
18 215
19 198
20 184
21 165
22 149
23 135
24 120
25 106
26 98
27 85
28 77
29 68
30 62
31 58
32 57
33 50
34 44
35 41
36 34
37 29
38 27
39 25
40 25
41 21
42 19
43 18
44 15
45 15
46 14
47 13
48 12
49 11
50 11
};
\addlegendentry{$q_1=0.02,q_2=0.001$}
\addplot [very thick, red, dashed]
table {%
0 20
0.331493532298907 30
0.53052675411486 36
0.750826443250065 46
0.924423939734709 56
1.03277506389981 64
1.15743273382902 72
1.24710854011159 80
1.33654319781191 90
1.4111095726043 100
1.47502337539424 110
1.52185247669072 120
1.55010573963622 126
1.59589244590701 136
1.62866626664228 144
1.65108743809961 152
1.67040072327266 158
1.71402235115418 168
1.74629394050721 178
1.77090900719386 186
1.80103442526454 194
1.82183469832074 202
1.87511471328231 212
1.89595607206142 218
1.93556536299896 228
1.95232109696891 236
1.99150427224896 244
2.03368336873053 252
2.06252132642165 260
2.09528661875333 268
2.11027234001918 278
2.15304412704448 288
2.17677071129489 296
2.20265416598587 302
2.21862804554337 312
2.23462434580032 320
2.24604878591273 328
2.26594029459036 336
2.28265253004172 344
2.31066631467568 354
2.32584418864688 360
2.35205152841355 364
2.36759712492755 374
2.38702410101235 380
2.41111674665054 388
2.42664932287756 396
2.44601284646137 404
2.47621465765882 410
2.49384810888473 418
2.51032996073797 428
2.5206683760362 434
2.53144094999465 442
2.54776264751399 450
2.56606198106071 454
2.58567013026086 462
2.62153103533729 472
2.66007522942526 482
2.68675250491289 488
2.71299443204206 494
2.72871935649562 502
2.74150484834842 504
2.75290493687134 512
2.77491870767085 520
2.7949002706074 526
2.8160604282265 534
2.8413402293713 542
2.8664025871402 548
2.89365877994456 556
2.91513608505286 564
2.9305689096493 574
2.95014743050437 582
2.97046549866513 588
2.98991773643515 598
3.01164800830593 604
3.04389508648865 608
3.0678236166049 618
3.10051144987905 626
3.11972670022738 632
3.13427451193604 638
3.1725069250987 648
3.19480835988703 654
3.23783639654127 664
3.25767180029696 666
3.26911629509097 672
3.30356128127026 680
3.34885410514319 682
3.367905842802 692
3.40160408024631 700
3.42411576749087 706
3.4558131960089 714
3.47239657365779 718
3.51035976064483 724
3.54320564518991 732
3.56612917025771 736
3.60150605087733 742
3.63876013537625 746
3.67731239691953 750
3.73085486499293 758
3.76260330694059 760
3.79358473274497 768
3.81022484636797 774
3.8434617120479 774
3.8728514877521 776
3.90978169556342 776
3.94863658647593 778
4.00202842713271 784
4.03166407601297 786
4.08634554664153 784
4.1545756547503 784
4.19787461786998 788
4.27362540883933 788
4.33990379659936 786
4.42136330213116 786
4.52662767017967 788
4.65280952101703 790
4.72256304782502 788
4.7981898364519 784
4.87045989479168 774
4.93831914290152 770
5.02537874614847 766
5.10372414349304 764
5.24371247784018 754
5.45099616822875 748
5.54714506302131 742
5.67338429007322 732
5.83771632497718 722
5.92150785234042 716
6.02831404612634 708
6.1486232574742 698
6.28399708744696 688
6.46417979801745 678
6.58833419953171 668
6.78340079520035 658
6.88548530447907 648
6.94655113957687 638
7.03093954915732 628
7.16259139162336 618
7.29812153013531 608
7.44590113598704 598
7.65218551165003 588
7.7628621458059 578
7.89088720664066 568
8.0611137590204 558
8.35126415328877 548
8.57373149036776 538
8.85017527648407 528
8.95515369457017 518
9.05870733269506 510
9.31695700639674 500
9.55151048470032 490
9.86866989388674 480
10.1380033176754 470
10.3877639057867 460
10.5502120127586 450
10.7404137143003 440
10.9135059257837 430
11.209136679819 420
11.4431792612654 410
11.6191154336605 400
12.0315309603443 390
12.2427818471688 380
12.6078259145981 370
12.8279173076926 360
13.072596921387 350
13.426981540034 340
13.7851402493 330
14.0595317206271 320
14.3543212793845 310
14.6810557357703 300
14.9106530712302 290
15.2371894092534 280
15.5413368881444 270
15.8111131991721 260
16.120466995295 250
16.3585549242916 240
16.9316419340308 230
17.2763766259891 220
17.5740078406511 210
18.291702632136 200
18.5664473846389 190
19.2224040112607 180
19.4666241638341 170
19.8465308547743 160
20.7111255482897 150
21.2898068686273 140
21.8447815827457 130
22.6145219398044 120
23.574287676135 110
24.3630183734139 100
25.5852897101321 90
26.1982744402657 80
27.3017160657094 70
29.4634569155705 60
31.8122751416149 50
34.7929363076566 40
36.7914770675387 30
40.4153050501606 20
45.2550483570333 10
};
\addlegendentry{$q_1=0.02,q_2=0.001$}

\addplot [very thick, blue]
table {%
0 20
1 42
2 62
3 86
4 126
5 193
6 277
7 367
8 466
9 544
10 577
11 583
12 575
13 557
14 515
15 471
16 430
17 387
18 356
19 319
20 286
21 255
22 229
23 209
24 182
25 164
26 144
27 129
28 120
29 105
30 96
31 83
32 79
33 74
34 71
35 62
36 55
37 50
38 44
39 37
40 32
41 29
42 26
43 24
44 23
45 23
46 16
47 16
48 14
49 13
50 13
};
\addlegendentry{$q_1=0.008,q_2=0.0004$}
\addplot [very thick, blue, dashed]
table {%
0 20
0.370003299814199 24
0.778757502632399 34
0.912356683691579 38
1.19849049775641 44
1.43167135926612 50
1.66573630938193 56
1.85734036750479 64
2.18557118974927 74
2.37898552695987 84
2.46449318429451 90
2.61200555788197 98
2.69613726045479 106
2.8318805994154 116
2.99492290548497 122
3.17474605454794 128
3.34195189643515 134
3.45390395459349 144
3.58221836353555 154
3.72452775213474 162
3.77647075896709 172
3.87637766663714 180
3.94222370487434 190
4.01159583944467 194
4.09007186604786 202
4.20776588407607 212
4.25875727313796 220
4.32454760601818 228
4.3603320261705 234
4.46133794231181 238
4.55711297559297 244
4.59505915293006 246
4.67250608162784 250
4.73112876595999 256
4.82285679988026 266
4.89287873871687 272
4.94919677963624 274
5.0091296084054 282
5.07325088116963 284
5.13335231236335 294
5.17460687513762 302
5.23348846500092 310
5.28032678457508 318
5.3312655017376 324
5.36546414635661 334
5.41208795408251 344
5.45249950378321 352
5.49560513340277 360
5.53738073668682 364
5.63088585156554 370
5.70587621360782 376
5.76159131352732 384
5.78852361466104 390
5.82397421657268 398
5.88148890544486 404
5.92991302342796 410
5.99197089816571 418
6.07022949766994 424
6.1263891537614 426
6.16843411386583 434
6.1990991947647 440
6.23815534207279 446
6.27711436342144 454
6.31385420205779 458
6.36245440161317 464
6.39809674749748 470
6.441334116327 476
6.51139725141181 476
6.57643735516591 482
6.60949761495584 488
6.68627133462741 490
6.7534685160528 496
6.79884485989947 500
6.84240635876559 502
6.89973040546961 504
6.97322008563374 506
7.01512546348796 510
7.09170400653443 516
7.16270836872648 524
7.21643690947861 526
7.30249897730422 528
7.37499967602842 534
7.42757064646757 534
7.47184116101918 536
7.50700378205092 538
7.55531035956635 542
7.60044285332062 544
7.65719159810065 546
7.71100266173482 550
7.74407089261824 558
7.81786817483493 564
7.87763635705008 570
7.95974283281004 572
8.03203029343656 576
8.07983323312935 576
8.15563370754265 580
8.18779486594346 584
8.28556258996157 588
8.34897722624977 594
8.39558515202567 596
8.47445230782631 600
8.54106594332965 608
8.63908476915463 604
8.68302753731337 608
8.74124717630454 608
8.82472263307358 606
8.87147363022952 614
8.9640097370456 614
9.02725698625544 616
9.12571439233657 622
9.21370216528695 620
9.27713627592115 622
9.34832022661024 622
9.44565890595858 616
9.50801147971873 616
9.60694377731679 620
9.66399295506758 618
9.71334647226201 624
9.80462540005669 622
9.91198804248287 620
10.0288826592213 620
10.1595487292502 616
10.2496125245033 614
10.3436585483596 612
10.4071585283551 604
10.511103777022 602
10.5884948370744 598
10.6668792880696 594
10.740521358061 588
10.8123711519939 586
10.9213893040196 578
11.0191884263455 572
11.1125858066784 566
11.3259774201119 562
11.4045771269447 552
11.5126096885064 554
11.6586121470518 548
11.8523886678884 540
11.9354044671641 532
12.0984505395552 524
12.2364355211421 518
12.4014130271622 510
12.5603512977986 504
12.7148640948966 498
12.7987438929607 492
12.946180448517 486
13.2416061872408 478
13.4610077203552 470
13.6089415613627 460
13.7851919911468 454
13.9932782617764 448
14.3167191894135 444
14.4941660604288 434
14.7837316373386 426
15.0159385596643 418
15.2824423316242 410
15.3763299371424 402
15.4904905678275 394
15.7665313224553 384
16.0465364350709 374
16.1835999301873 364
16.4966223962393 354
16.8841788237923 346
17.0759731525476 336
17.2860250587344 326
17.5236207053978 316
17.8548911584964 308
18.0798863926054 298
18.3159691606456 288
18.5934558365468 280
18.9771118861988 270
19.4359945020895 264
19.8883865979067 254
20.1485404900136 244
20.6822921864954 234
20.9318160189133 224
21.16868015285 214
21.5788531271879 204
22.1640790835238 194
22.539672337592 186
22.7671515267288 176
23.2076590173971 168
24.0408513409017 158
24.5716901945803 148
25.3014911088296 138
26.4651160670777 128
27.5311521949247 118
28.2959517230841 108
30.1734072894204 98
31.257318349694 88
32.2746913538974 78
33.3387679394026 68
35.5399800392832 58
36.6286759040803 48
38.8869689197671 38
41.6797457017655 28
46.0387378809054 18
};
\addlegendentry{$q_1=0.008,q_2=0.0004$}

\addplot [very thick,  black!50!green]
table {%
0 20.61
1 27.56
2 36.92
3 49.04
4 63.72
5 81.37
6 102
7 126.59
8 154.85
9 185.21
10 217.39
11 248.95
12 279.7
13 308.77
14 333.08
15 355.29
16 371.85
17 383.08
18 388.23
19 390.09
20 386.04
21 376.7
22 366.21
23 350.98
24 335.23
25 316.94
26 297.96
27 278.92
28 258.86
29 239.73
30 221.49
31 204.13
32 188.04
33 172.3
34 158.31
35 145.04
36 132.61
37 120.99
38 110.16
39 101
40 92
41 83.92
42 76.28
43 69.81
44 63.52
45 58.06
46 52.32
47 47.57
48 42.97
49 38.92
50 34.8
};
\addlegendentry{$q_1=0.008,q_2=0.0001$}
\addplot [very thick, black!50!green, dashed]
table {%
0 20
1 23.71
2 27.82
3 32.42
4 37.53
5 43.18
6 50.35
7 59.58
8 70.98
9 84.98
10 100.83
11 119.98
12 141.68
13 166.91
14 195.78
15 225.6
16 253.73
17 281.71
18 306.11
19 328.13
20 346.28
21 359.62
22 367.37
23 370.2
24 368.33
25 362.43
26 353.17
27 341.05
28 327.35
29 312.23
30 295.79
31 278.61
32 261.29
33 244.13
34 227.2
35 211.02
36 194.65
37 180.04
38 165.86
39 151.87
40 139.28
41 127.3
42 116.32
43 106.55
44 97.31
45 89.11
46 81.08
47 74.59
48 67.68
49 61.66
50 56.31
};
\addlegendentry{$q_1=0.008,q_2=0.0001$}

\end{axis}

\end{tikzpicture}

%% file: Figures/heuristic_1.tex
\begin{tikzpicture}

\definecolor{color0}{rgb}{0.12156862745098,0.466666666666667,0.705882352941177}
\definecolor{color1}{rgb}{1,0.498039215686275,0.0549019607843137}
\definecolor{color2}{rgb}{0.172549019607843,0.627450980392157,0.172549019607843}
\definecolor{color3}{rgb}{0.83921568627451,0.152941176470588,0.156862745098039}
\definecolor{color4}{rgb}{0.580392156862745,0.403921568627451,0.741176470588235}
\definecolor{color5}{rgb}{0.549019607843137,0.337254901960784,0.294117647058824}

\begin{groupplot}[group style={group size=2 by 1, horizontal sep=1.2cm}]
\nextgroupplot[
legend cell align={left},
legend style={fill opacity=0.8, draw opacity=1, text opacity=1, draw=white!80!black},
tick align=outside,
tick pos=left,
title={Average number of infected individuals},
x grid style={white!69.0196078431373!black},
xlabel={Day},
xmajorgrids,
xmin=-2.5, xmax=52.5,
xminorgrids,
xtick style={color=black},
y grid style={white!69.0196078431373!black},
ymajorgrids,
ymin=-31.0767, ymax=685.8747,
yminorgrids,
ytick style={color=black}
]
\addplot [very thick, color0]
table {%
0 19.688
1 35.898
2 64.962
3 113.984
4 190.712
5 297.756
6 423.926
7 542.932
8 622.988
9 653.286
10 641.47
11 605.706
12 559.956
13 511.474
14 465.08
15 420.98
16 380.408
17 343.174
18 309.464
19 278.972
20 251.244
21 226.698
22 204.316
23 183.956
24 165.584
25 149.206
26 134.54
27 120.82
28 108.602
29 98.008
30 88.08
31 79.336
32 71.45
33 64.374
34 58.146
35 52.358
36 47.166
37 42.556
38 38.338
39 34.516
40 30.984
41 27.952
42 25.216
43 22.802
44 20.494
45 18.528
46 16.602
47 14.958
48 13.478
49 12.206
50 10.93
};
\addlegendentry{No testing}
\addplot [very thick, color1]
table {%
0 19.844
1 36.866
2 50.506
3 60.688
4 67.892
5 72.698
6 74.88
7 75.342
8 74.468
9 72.388
10 69.638
11 65.978
12 61.944
13 57.724
14 53.588
15 49.416
16 45.582
17 41.834
18 38.35
19 34.918
20 31.816
21 29.086
22 26.47
23 24.03
24 21.74
25 19.608
26 17.73
27 16.052
28 14.588
29 13.316
30 12.076
31 10.852
32 9.756
33 8.776
34 7.94
35 7.128
36 6.44
37 5.832
38 5.274
39 4.78
40 4.29
41 3.908
42 3.53
43 3.15
44 2.888
45 2.558
46 2.308
47 2.114
48 1.916
49 1.712
50 1.522
};
\addlegendentry{Complete}
\addplot [very thick, color2]
table {%
0 20.138
1 36.786
2 50.104
3 60.64
4 67.888
5 72.174
6 74.382
7 75.144
8 73.89
9 71.656
10 68.57
11 65.348
12 61.944
13 58.564
14 55.418
15 52.306
16 50.144
17 48.31
18 47.546
19 47.416
20 48.066
21 49.532
22 51.92
23 54.638
24 57.5
25 60.22
26 62.712
27 65.182
28 66.488
29 67.27
30 67.2
31 66.85
32 65.76
33 63.988
34 62.008
35 59.492
36 56.658
37 53.486
38 50.076
39 46.44
40 43.004
41 39.374
42 35.832
43 32.538
44 29.374
45 26.472
46 23.73
47 21.404
48 19.256
49 17.416
50 15.628
};
\addlegendentry{Rnd. Grp. mean}
\addplot [very thick, color3]
table {%
0 20.044
1 37.156
2 50.912
3 61.256
4 68.392
5 73.52
6 76.29
7 77.21
8 76.246
9 74.1
10 71.12
11 67.462
12 63.558
13 59.478
14 55.374
15 51.344
16 47.676
17 44.758
18 42.012
19 39.952
20 38.126
21 36.702
22 35.63
23 34.582
24 33.758
25 33.276
26 32.674
27 32.062
28 31.148
29 29.942
30 28.506
31 26.754
32 24.854
33 22.866
34 20.964
35 19.264
36 17.73
37 16.19
38 14.834
39 13.61
40 12.654
41 11.784
42 11.018
43 10.324
44 9.668
45 8.968
46 8.288
47 7.608
48 6.968
49 6.316
50 5.658
};
\addlegendentry{CCA}
\addplot [very thick, color4]
table {%
0 20.39
1 37.564
2 51.124
3 61.39
4 68.644
5 73.288
6 75.706
7 76.022
8 74.904
9 72.664
10 69.478
11 65.798
12 61.982
13 57.906
14 53.73
15 49.644
16 45.866
17 41.95
18 38.402
19 35.054
20 31.982
21 29.054
22 26.322
23 23.902
24 21.548
25 19.464
26 17.618
27 15.988
28 14.446
29 13.052
30 11.77
31 10.628
32 9.576
33 8.66
34 7.854
35 7.016
36 6.3
37 5.632
38 5.056
39 4.61
40 4.194
41 3.766
42 3.372
43 3.054
44 2.792
45 2.518
46 2.282
47 2.072
48 1.864
49 1.672
50 1.512
};
\addlegendentry{Rnd. Grp. max. }

\nextgroupplot[
legend cell align={left},
legend style={fill opacity=0.8, draw opacity=1, text opacity=1, at={(0.91,0.5),}, anchor=east, draw=white!80!black},
tick align=outside,
tick pos=left,
title={Average number of tests used},
x grid style={white!69.0196078431373!black},
xlabel={Day},
xmajorgrids,
xmin=-2.45, xmax=51.45,
xminorgrids,
xtick style={color=black},
y grid style={white!69.0196078431373!black},
ymajorgrids,
ymin=-50, ymax=1050,
yminorgrids,
ytick style={color=black}
]
\addplot [very thick, color0]
table {%
0 0
1 0
2 0
3 0
4 0
5 0
6 0
7 0
8 0
9 0
10 0
11 0
12 0
13 0
14 0
15 0
16 0
17 0
18 0
19 0
20 0
21 0
22 0
23 0
24 0
25 0
26 0
27 0
28 0
29 0
30 0
31 0
32 0
33 0
34 0
35 0
36 0
37 0
38 0
39 0
40 0
41 0
42 0
43 0
44 0
45 0
46 0
47 0
48 0
49 0
};
\addlegendentry{No testing}
\addplot [very thick, color1]
table {%
0 1000
1 980.156
2 961.174
3 943.95
4 928.78
5 915.45
6 903.92
7 894.302
8 886.472
9 879.886
10 874.558
11 870.366
12 867.074
13 864.518
14 862.402
15 860.688
16 859.43
17 858.416
18 857.65
19 857.008
20 856.488
21 856.042
22 855.668
23 855.394
24 855.196
25 855.04
26 854.874
27 854.76
28 854.638
29 854.562
30 854.502
31 854.462
32 854.438
33 854.418
34 854.404
35 854.396
36 854.39
37 854.39
38 854.39
39 854.39
40 854.39
41 854.39
42 854.39
43 854.39
44 854.39
45 854.39
46 854.39
47 854.39
48 854.39
49 854.39
};
\addlegendentry{Complete}
\addplot [very thick, color2]
table {%
0 1000
1 979.862
2 960.276
3 936.022
4 910.116
5 873.238
6 811.372
7 749.926
8 682.614
9 608.17
10 524.382
11 449.484
12 380.534
13 315.018
14 257.228
15 196.854
16 153.008
17 117.518
18 89.074
19 70.472
20 56.578
21 44.798
22 33.838
23 24.106
24 16.998
25 10.686
26 6.448
27 5.102
28 3.378
29 2.402
30 1.808
31 1.82
32 2.39
33 1.278
34 1
35 1
36 1
37 1
38 1
39 1
40 1
41 1
42 1
43 1
44 1
45 1
46 1
47 1
48 1
49 1
};
\addlegendentry{Rnd. Grp. mean}
\addplot [very thick, color3]
table {%
0 1000
1 979.956
2 960.216
3 933.352
4 908.518
5 870.272
6 820.702
7 769.736
8 709.51
9 636.076
10 555.688
11 468.498
12 385.686
13 322.688
14 257.596
15 202.536
16 163.298
17 128.978
18 106.248
19 82.242
20 55.408
21 44.77
22 33.4
23 24.298
24 20.822
25 18.704
26 15.064
27 12.318
28 11.324
29 10.514
30 8.058
31 6.504
32 3.49
33 2.634
34 2.396
35 2.91
36 2.85
37 2.404
38 2.06
39 1.78
40 1.286
41 1.602
42 1
43 1.61
44 1.29
45 1
46 1.294
47 1
48 1
49 1.29
};
\addlegendentry{CCA}
\addplot [very thick, color4]
table {%
0 1000
1 979.61
2 960.078
3 936.702
4 906.024
5 869.646
6 825.096
7 771.162
8 697.742
9 612.33
10 526.574
11 441.91
12 369.378
13 303.774
14 258.896
15 209.882
16 165.85
17 130.788
18 99.514
19 76.654
20 57.298
21 48.694
22 33.734
23 26.076
24 19.88
25 14.222
26 12.144
27 10.722
28 9.824
29 8.238
30 7.306
31 4.662
32 2.858
33 3.594
34 3.35
35 3.268
36 2.036
37 1.96
38 1.992
39 1
40 1
41 1
42 1
43 1
44 1
45 1
46 1
47 1
48 1
49 1
};
\addlegendentry{Rnd. Grp. max.}
\addplot [very thick, color5]
table {%
0 141.440542541821
1 128.245110424809
2 118.772938112971
3 106.531232263447
4 93.6579791075399
5 81.9747940538131
6 70.9775526120461
7 59.7927748918724
8 49.0718824448902
9 41.3843982547259
10 33.9219061968275
11 27.038252975265
12 21.4440109074949
13 16.8037390768501
14 13.8514364092336
15 11.2589580772407
16 8.31101223023899
17 6.75036614190752
18 5.15108318185189
19 4.27190322770539
20 3.49148857924924
21 2.96154072034629
22 2.49914210571421
23 1.81683848628215
24 1.35464550730173
25 1.07092519324101
26 1.09650201133368
27 0.75986193267728
28 0.807125635884263
29 0.514097167699905
30 0.41043416369512
31 0.270451384321691
32 0.164049538307259
33 0.130889904056867
34 0.0898658476212797
35 0.0532526576864038
36 0.0413125365794689
37 0
38 0
39 0
40 0
41 0
42 0
43 0
44 0
45 0
46 0
47 0
48 0
49 0
};
\addlegendentry{Lower bound}
\end{groupplot}

\end{tikzpicture}

%% file: Figures/heuristic_2.tex
\begin{tikzpicture}

\definecolor{color0}{rgb}{0.12156862745098,0.466666666666667,0.705882352941177}
\definecolor{color1}{rgb}{1,0.498039215686275,0.0549019607843137}
\definecolor{color2}{rgb}{0.172549019607843,0.627450980392157,0.172549019607843}
\definecolor{color3}{rgb}{0.83921568627451,0.152941176470588,0.156862745098039}
\definecolor{color4}{rgb}{0.580392156862745,0.403921568627451,0.741176470588235}
\definecolor{color5}{rgb}{0.549019607843137,0.337254901960784,0.294117647058824}

\begin{groupplot}[group style={group size=2 by 1, horizontal sep=1.3cm}]
\nextgroupplot[
legend cell align={left},
legend style={fill opacity=0.8, draw opacity=1, text opacity=1, draw=white!80!black},
tick align=outside,
tick pos=left,
title={Average number of infected individuals},
x grid style={white!69.0196078431373!black},
xlabel={Day},
xmajorgrids,
xmin=-2.5, xmax=52.5,
xminorgrids,
xtick style={color=black},
y grid style={white!69.0196078431373!black},
ymajorgrids,
ymin=-31.0767, ymax=3500,
yminorgrids,
ytick style={color=black}
]
\addplot [very thick, color0]
table {%
0 99.936
1 185.874
2 338.446
3 596.43
4 1001.412
5 1566.986
6 2229.278
7 2827.064
8 3198.75
9 3302.87
10 3207.55
11 3008.992
12 2768.57
13 2523.018
14 2287.84
15 2068.59
16 1867.182
17 1683.94
18 1518.514
19 1367.844
20 1231.738
21 1110.522
22 999.766
23 901.166
24 811.192
25 730.406
26 656.778
27 591.476
28 532.142
29 478.76
30 431.4
31 387.566
32 348.808
33 314.04
34 282.996
35 254.68
36 228.966
37 206.098
38 185.6
39 166.98
40 150.158
41 134.954
42 121.566
43 109.384
44 98.466
45 88.516
46 79.694
47 71.8
48 64.45
49 57.978
50 52.08
};
\addlegendentry{No testing}
\addplot [very thick, color1]
table {%
0 99.758
1 185.862
2 256.074
3 310.836
4 350.648
5 377.262
6 392.304
7 398.242
8 396.208
9 388.062
10 374.746
11 358.768
12 339.834
13 320.008
14 299.822
15 279.58
16 259.47
17 239.762
18 221.018
19 203.182
20 186.366
21 170.708
22 155.998
23 142.362
24 129.738
25 117.938
26 107.234
27 97.09
28 88.096
29 79.71
30 72.138
31 65.108
32 58.96
33 53.328
34 48.186
35 43.444
36 39.04
37 35.238
38 31.716
39 28.554
40 25.81
41 23.276
42 20.936
43 18.944
44 16.934
45 15.312
46 13.842
47 12.452
48 11.176
49 10.1
50 9.092
};
\addlegendentry{Complete}
\addplot [very thick, color2]
table {%
0 100.644
1 187.142
2 257.7
3 312.422
4 352.832
5 378.782
6 393.238
7 399.016
8 397.228
9 389.01
10 376.624
11 360.79
12 342.634
13 323.02
14 302.188
15 281.964
16 261.75
17 241.996
18 223.086
19 205.398
20 188.4
21 172.654
22 157.904
23 144.822
24 133.778
25 124.528
26 117.93
27 114.404
28 113.608
29 115.754
30 120.564
31 127.25
32 135.936
33 146.194
34 158.624
35 172.43
36 187.194
37 202.662
38 217.16
39 229.906
40 238.896
41 244.544
42 245.52
43 242.568
44 236.486
45 229.55
46 221.164
47 212.022
48 202.076
49 191.068
50 179.62
};
\addlegendentry{Rnd. Grp. mean}
\addplot [very thick, color3]
table {%
0 99.974
1 185.386
2 254.832
3 309.326
4 348.7
5 375.294
6 390.694
7 396.302
8 393.774
9 386.012
10 373.568
11 357.918
12 339.906
13 320.022
14 299.522
15 279.508
16 259.124
17 239.458
18 220.312
19 202.346
20 185.662
21 169.874
22 155.274
23 142.062
24 129.638
25 118.766
26 109.072
27 101.068
28 94.454
29 89.432
30 86.12
31 83.9
32 82.702
33 82.054
34 81.692
35 81.066
36 80.11
37 78.35
38 76.106
39 73.308
40 70.686
41 68.24
42 66.098
43 64.34
44 62.748
45 60.938
46 58.276
47 55.336
48 51.768
49 47.956
50 43.958
};
\addlegendentry{CCA}
\addplot [very thick, color4]
table {%
0 100.148
1 185.3
2 254.496
3 308.276
4 347.824
5 373.722
6 388.864
7 394.484
8 391.602
9 382.916
10 370.388
11 354.91
12 336.36
13 316.342
14 296.312
15 276.13
16 256.042
17 236.794
18 218.426
19 200.724
20 184.12
21 168.808
22 154.398
23 141.396
24 130.008
25 119.92
26 112.264
27 106.716
28 104.764
29 106.098
30 111.602
31 122.192
32 136.044
33 153.352
34 173.71
35 196.504
36 221.992
37 248.23
38 274.162
39 297.544
40 315.79
41 327.936
42 333.7
43 332.946
44 325.092
45 311.794
46 294.554
47 274.762
48 254.21
49 233.812
50 215.048
};
\addlegendentry{Rnd. Grp. max. }

\nextgroupplot[
legend cell align={left},
legend style={fill opacity=0.8, draw opacity=1, text opacity=1, at={(0.91,0.5)}, anchor=east, draw=white!80!black},
tick align=outside,
tick pos=left,
title={Average number of tests used},
x grid style={white!69.0196078431373!black},
xlabel={Day},
xmajorgrids,
xmin=-2.45, xmax=51.45,
xminorgrids,
xtick style={color=black},
y grid style={white!69.0196078431373!black},
ymajorgrids,
ymin=-50, ymax=5200,
yminorgrids,
ytick style={color=black}
]
\addplot [very thick, color0]
table {%
0 0
1 0
2 0
3 0
4 0
5 0
6 0
7 0
8 0
9 0
10 0
11 0
12 0
13 0
14 0
15 0
16 0
17 0
18 0
19 0
20 0
21 0
22 0
23 0
24 0
25 0
26 0
27 0
28 0
29 0
30 0
31 0
32 0
33 0
34 0
35 0
36 0
37 0
38 0
39 0
40 0
41 0
42 0
43 0
44 0
45 0
46 0
47 0
48 0
49 0
};
\addlegendentry{No testing}
\addplot [very thick, color1]
table {%
0 5000
1 4900.242
2 4804.178
3 4715.234
4 4634.84
5 4564.156
6 4502.334
7 4449.342
8 4404.258
9 4366.448
10 4335.424
11 4309.474
12 4288.324
13 4270.896
14 4256.686
15 4244.826
16 4235.284
17 4227.504
18 4221.108
19 4215.968
20 4211.632
21 4208.226
22 4205.488
23 4203.298
24 4201.43
25 4199.968
26 4198.782
27 4197.862
28 4197.058
29 4196.404
30 4195.924
31 4195.508
32 4195.158
33 4194.89
34 4194.698
35 4194.534
36 4194.418
37 4194.318
38 4194.222
39 4194.152
40 4194.118
41 4194.092
42 4194.08
43 4194.066
44 4194.056
45 4194.052
46 4194.048
47 4194.046
48 4194.044
49 4194.044
};
\addlegendentry{Complete}
\addplot [very thick, color2]
table {%
0 5000
1 4899.356
2 4802.728
3 4713.804
4 4633.462
5 4561.89
6 4498.982
7 4436.912
8 4375.34
9 4251.28
10 4032.568
11 3722.176
12 3368.284
13 2956.194
14 2449.074
15 2043.174
16 1673.202
17 1348.874
18 1116.036
19 918.608
20 724.03
21 555.258
22 423.958
23 337.91
24 262.944
25 206.502
26 165.684
27 132.43
28 103.164
29 91.232
30 68.668
31 54.138
32 35.666
33 26.06
34 15.624
35 11.78
36 7.214
37 7.204
38 6.56
39 4.018
40 3.52
41 2.05
42 1.374
43 1.368
44 1
45 1
46 1
47 1
48 1
49 1
};
\addlegendentry{Rnd. Grp. mean}
\addplot [very thick, color3]
table {%
0 5000
1 4900.026
2 4804.656
3 4716.652
4 4636.936
5 4566.376
6 4504.574
7 4445.93
8 4375.78
9 4270.448
10 4101.69
11 3763.208
12 3391.976
13 2926.742
14 2431.536
15 2033.888
16 1648.688
17 1336.718
18 1084.018
19 851.204
20 702.77
21 567.26
22 475.464
23 387.354
24 311.506
25 231.12
26 170.514
27 135.968
28 90.936
29 77.124
30 61.89
31 47.714
32 38.046
33 27.386
34 28.154
35 20.844
36 15.53
37 12.12
38 10.484
39 9.762
40 8.534
41 7.552
42 4.002
43 2.264
44 3.056
45 1.67
46 1.352
47 1.366
48 1.352
49 1
};
\addlegendentry{CCA}
\addplot [very thick, color4]
table {%
0 5000
1 4899.852
2 4804.578
3 4716.962
4 4637.62
5 4567.674
6 4506.758
7 4443.648
8 4370.408
9 4245.704
10 4052.044
11 3756.926
12 3307.594
13 2852.178
14 2397.238
15 2054.836
16 1676.798
17 1380.916
18 1127.16
19 906.048
20 714.95
21 577.466
22 477.706
23 368.672
24 279.59
25 228.602
26 184.53
27 145.266
28 107.72
29 89.708
30 72.632
31 61.664
32 43.56
33 36.296
34 32.852
35 25.538
36 19.45
37 16.154
38 10.562
39 10.09
40 10.434
41 7.644
42 5.562
43 5.812
44 6.69
45 6.022
46 5.676
47 5.544
48 3.754
49 2.684
};
\addlegendentry{Rnd. Grp. max.}
\addplot [very thick, color5]
table {%
0 707.202712709101
1 653.973382101686
2 610.360154018117
3 557.717691327573
4 502.578628744084
5 443.600981071981
6 390.391425248199
7 337.922431191269
8 290.943097518855
9 247.29315251692
10 206.263348499004
11 175.194398463947
12 145.242034809726
13 121.424726441868
14 100.42330560151
15 84.6785336141085
16 68.9405662591154
17 56.6677420632803
18 46.9194843429803
19 38.2324020598705
20 32.2686666574364
21 25.5983334652033
22 20.7633885616763
23 16.822112159269
24 14.2316581103508
25 11.2279084869799
26 9.12234963717444
27 7.11630729165668
28 6.21289922233003
29 5.05947885012259
30 3.7447167028042
31 3.22523072767378
32 2.71528969019449
33 2.0385903179288
34 1.52896511455615
35 1.28470723546018
36 0.907841802486064
37 0.797173589978574
38 0.752777961812139
39 0.540455288102603
40 0.277569220796371
41 0.205550785126329
42 0.0986672931647189
43 0.112201052971036
44 0.0789859162676169
45 0.0335839758292738
46 0.0343112404877179
47 0.0166723592816167
48 0.0164848034619211
49 0
};
\addlegendentry{Lower bound}
\end{groupplot}

\end{tikzpicture}

%% file: main.bbl
\begin{thebibliography}{10}
\providecommand{\url}[1]{#1}
\csname url@samestyle\endcsname
\providecommand{\newblock}{\relax}
\providecommand{\bibinfo}[2]{#2}
\providecommand{\BIBentrySTDinterwordspacing}{\spaceskip=0pt\relax}
\providecommand{\BIBentryALTinterwordstretchfactor}{4}
\providecommand{\BIBentryALTinterwordspacing}{\spaceskip=\fontdimen2\font plus
\BIBentryALTinterwordstretchfactor\fontdimen3\font minus
  \fontdimen4\font\relax}
\providecommand{\BIBforeignlanguage}[2]{{%
\expandafter\ifx\csname l@#1\endcsname\relax
\typeout{** WARNING: IEEEtran.bst: No hyphenation pattern has been}%
\typeout{** loaded for the language `#1'. Using the pattern for}%
\typeout{** the default language instead.}%
\else
\language=\csname l@#1\endcsname
\fi
#2}}
\providecommand{\BIBdecl}{\relax}
\BIBdecl

\bibitem{art1}
C.~Gollier and O.~Gossner, ``Group testing against covid-19,'' April 2020, see
  \url{https://www.tse-fr.eu/articles/group-testing-against-covid-19}.

\bibitem{art2}
M.~Broadfoot, ``Coronavirus test shortages trigger a new strategy: Group
  screening,'' May 2020, see
  \url{https://www.scientificamerican.com/article/coronavirus-test-shortages-trigger-a-new-strategy-group-screening2/}.

\bibitem{art4}
J.~Ellenberg, ``Five people. one test. this is how you get there.''
  \emph{NYtimes}, May 2020.

\bibitem{Cov-GpTest-1}
C.~Verdun \emph{et~al.}, ``Group testing for sars-cov-2 allows up to 10-fold
  efficiency increase across realistic scenarios and testing strategies,''
  \emph{medRxiv}, 2020.

\bibitem{Cov-GpTest-2}
S.~Ghosh \emph{et~al.}, ``Tapestry: A single-round smart pooling technique for
  covid-19 testing,'' \emph{medRxiv}, 2020.

\bibitem{kucirka2020-PCR}
L.~M. Kucirka, S.~A. Lauer, O.~Laeyendecker, D.~Boon, and J.~Lessler,
  ``Variation in false-negative rate of reverse transcriptase polymerase chain
  reaction--based sars-cov-2 tests by time since exposure,'' \emph{Annals of
  Internal Medicine}, vol. 173, pp. 262--267, Aug. 2020.

\bibitem{taipale2021populationscale}
J.~Taipale, I.~Kontoyiannis, and S.~Linnarsson, ``Population-scale testing can
  suppress the spread of infectious disease,'' 2021.

\bibitem{taipale2020population}
J.~Taipale, P.~Romer, and S.~Linnarsson, ``Population-scale testing can
  suppress the spread of covid-19,'' \emph{MedRxiv}, 2020.

\bibitem{bergstrom2020frequency}
T.~Bergstrom, C.~T. Bergstrom, and H.~Li, ``Frequency and accuracy of proactive
  testing for covid-19,'' \emph{medRxiv}, 2020.

\bibitem{GroupTestingMonograph}
M.~Aldridge, O.~Johnson, and J.~Scarlett, ``Group testing: an information
  theory perspective,'' \emph{CoRR}, vol. abs/1902.06002, 2019.

\bibitem{isit-paper}
S.~R. Srinivasavaradhan, P.~Nikolopoulos, C.~Fragouli, and S.~Diggavi, ``An
  entropy reduction approach to continual testing,'' \emph{Accepted to appear
  at IEEE ISIT}, 2021.

\bibitem{mathOfEpidemicsOnNetworks}
I.~Kiss, J.~Miller, and P.~Simon, \emph{Mathematics of Epidemics on Networks},
  01 2017, vol.~46.

\bibitem{prior}
T.~{Li}, C.~L. {Chan}, W.~{Huang}, T.~{Kaced}, and S.~{Jaggi}, ``Group testing
  with prior statistics,'' in \emph{2014 IEEE International Symposium on
  Information Theory}, 2014, pp. 2346--2350.

\bibitem{kempe2003maximizing}
D.~Kempe, J.~Kleinberg, and {\'E}.~Tardos, ``Maximizing the spread of influence
  through a social network,'' in \emph{Proceedings of the ninth ACM SIGKDD
  international conference on Knowledge discovery and data mining}, 2003, pp.
  137--146.

\bibitem{goenka2020contact}
R.~Goenka, S.-J. Cao, C.-W. Wong, A.~Rajwade, and D.~Baron, ``Contact tracing
  enhances the efficiency of covid-19 group testing,'' \emph{arXiv preprint
  arXiv:2011.14186}, 2020.

\bibitem{molnar2020safety}
T.~G. Moln{\'a}r, A.~W. Singletary, G.~Orosz, and A.~D. Ames, ``Safety-critical
  control of compartmental epidemiological models with measurement delays,''
  \emph{IEEE Control Systems Letters}, vol.~5, no.~5, pp. 1537--1542, 2020.

\bibitem{GroupTesting-community-nonOverlap}
P.~Nikolopoulos, S.~Rajan~Srinivasavaradhan, T.~Guo, C.~Fragouli, and
  S.~Diggavi, ``Group testing for connected communities,'' in \emph{Proceedings
  of The 24th International Conference on Artificial Intelligence and
  Statistics}, vol. 130.\hskip 1em plus 0.5em minus 0.4em\relax PMLR, 2021, pp.
  2341--2349. See also arXiv preprint arXiv:2007.08\,111.

\bibitem{GroupTesting-community-overlap}
P.~Nikolopoulos, S.~R. Srinivasavaradhan, T.~Guo, C.~Fragouli, and S.~Diggavi,
  ``Group testing for overlapping communities,'' \emph{arXiv preprint
  arXiv:2012.02804, IEEE ICC}, 2021.

\bibitem{zhu2020noisy}
J.~Zhu, K.~Rivera, and D.~Baron, ``Noisy pooled pcr for virus testing,''
  \emph{arXiv preprint arXiv:2004.02689}, 2020.

\bibitem{ayfer2021adaptive}
S.~Ahn, W.-N. Chen, and A.~Ozgur, ``Adaptive group testing on networks with
  community structure,'' \emph{arXiv preprint arXiv:2101.02405}.

\bibitem{sennur2021group}
B.~Arasli and S.~Ulukus, ``Group testing with a graph infection spread model,''
  \emph{arXiv preprint arXiv:2101.05792}, 2021.

\bibitem{bertolotti2020network}
P.~Bertolotti and A.~Jadbabaie, ``Network group testing,'' \emph{arXiv preprint
  arXiv:2012.02847}, 2020.

\bibitem{cheraghchi2012graph}
M.~Cheraghchi, A.~Karbasi, S.~Mohajer, and V.~Saligrama, ``Graph-constrained
  group testing,'' \emph{IEEE Transactions on Information Theory}, vol.~58,
  no.~1, pp. 248--262, 2012.

\bibitem{karbasi2012sequential}
A.~Karbasi and M.~Zadimoghaddam, ``Sequential group testing with graph
  constraints,'' in \emph{2012 IEEE information theory workshop}.\hskip 1em
  plus 0.5em minus 0.4em\relax Ieee, 2012, pp. 292--296.

\bibitem{bay2020optimal}
W.~H. Bay, E.~Price, and J.~Scarlett, ``Optimal non-adaptive probabilistic
  group testing requires $\theta(\min\{k \log n, n\})$ tests,'' 2020.

\bibitem{coja-oghlan20a}
A.~Coja-Oghlan, O.~Gebhard, M.~Hahn-Klimroth, and P.~Loick, ``Optimal group
  testing,'' ser. Proceedings of Machine Learning Research, J.~Abernethy and
  S.~Agarwal, Eds., vol. 125, Jul. 2020, pp. 1374--1388.

\bibitem{price2020fast}
E.~Price and J.~Scarlett, ``A fast binary splitting approach to non-adaptive
  group testing,'' \emph{arXiv preprint arXiv:2006.10268}, 2020.

\bibitem{bernoulli_testing1}
M.~Aldridge, L.~Baldassini, and O.~Johnson, ``Group testing algorithms: Bounds
  and simulations,'' \emph{IEEE Transactions on Information Theory}, vol.~60,
  no.~6, pp. 3671--3687, 2014.

\bibitem{bernoulli_testing2}
G.~K. Atia and V.~Saligrama, ``Boolean compressed sensing and noisy group
  testing,'' \emph{IEEE Transactions on Information Theory}, vol.~58, no.~3,
  pp. 1880--1901, 2012.

\bibitem{PhaseTrans-SODA16}
J.~Scarlett and V.~Cevher, ``Phase transitions in group testing,'' in
  \emph{Proceedings of the Twenty-Seventh Annual {ACM-SIAM} Symposium on
  Discrete Algorithms, {SODA} 2016, Arlington, VA, USA, January 10-12,
  2016}.\hskip 1em plus 0.5em minus 0.4em\relax {SIAM}, 2016, pp. 40--53.

\bibitem{coja-oghlan19}
A.~{Coja-Oghlan}, O.~{Gebhard}, M.~{Hahn-Klimroth}, and P.~{Loick},
  ``Information-theoretic and algorithmic thresholds for group testing,''
  \emph{IEEE Trans. Inf. Theory}, 2020.

\bibitem{ncc-Johnson}
O.~{Johnson}, M.~{Aldridge}, and J.~{Scarlett}, ``Performance of group testing
  algorithms with near-constant tests per item,'' \emph{IEEE Trans. Inf.
  Theory}, vol.~65, no.~2, pp. 707--723, 2019.

\bibitem{individual-optimal}
M.~{Aldridge}, ``Individual testing is optimal for nonadaptive group testing in
  the linear regime,'' \emph{IEEE Trans. Inf. Theory}, vol.~65, no.~4, 2019.

\bibitem{Nonadaptive-1}
C.~L. Chan, S.~Jaggi, V.~Saligrama, and S.~Agnihotri, ``Non-adaptive group
  testing: Explicit bounds and novel algorithms,'' \emph{IEEE Trans. Inf.
  Theory}, vol.~60, no.~5, p. 3019–3035, 2014.

\bibitem{ac-dc}
R.~Gabrys, S.~Pattabiraman, V.~Rana, J.~Ribeiro, M.~Cheraghchi, V.~Guruswami,
  and O.~Milenkovic, ``Ac-dc: Amplification curve diagnostics for covid-19
  group testing,'' 2020.

\bibitem{bondorf2020sublinear}
S.~Bondorf, B.~Chen, J.~Scarlett, H.~Yu, and Y.~Zhao, ``Sublinear-time
  non-adaptive group testing with o (klogn) tests via bit-mixing coding,''
  \emph{IEEE Transactions on Information Theory}, 2020.

\bibitem{Nonadaptive-2}
S.~Cai, M.~Jahangoshahi, M.~Bakshi, and S.~Jaggi, ``Efficient algorithms for
  noisy group testing,'' \emph{IEEE Trans. Inf. Theory}, vol.~63, no.~4, p.
  2113–2136, 2017.

\bibitem{Nonadaptive-3}
\BIBentryALTinterwordspacing
H.~A. Inan, P.~Kairouz, M.~Wootters, and A.~{\"{O}}zg{\"{u}}r, ``On the
  optimality of the kautz-singleton construction in probabilistic group
  testing,'' \emph{CoRR}, vol. abs/1808.01457, 2018. [Online]. Available:
  \url{http://arxiv.org/abs/1808.01457}
\BIBentrySTDinterwordspacing

\bibitem{Saffron}
K.~Lee, R.~Pedarsani, and K.~Ramchandran, ``Saffron: A fast, efficient, and
  robust framework for group testing based on sparse-graph codes,'' in
  \emph{2016 IEEE International Symposium on Information Theory (ISIT)}, 2016,
  pp. 2873--2877.

\end{thebibliography}
